\documentclass[article,a4paper,oneside,10pt]{memoir}

\usepackage{amsmath}
\usepackage{amssymb}
\usepackage{graphicx}

\usepackage{amsthm}
\usepackage{amstext}
\usepackage{graphicx}

\usepackage{color}

\usepackage{pict2e}

\usepackage[pdftex,plainpages=false,colorlinks,hyperindex,bookmarksopen,linkcolor=gray,citecolor=blue,urlcolor=blue]{hyperref}

\usepackage{tikz}
\usetikzlibrary{shapes}
\usetikzlibrary{positioning}

\usepackage{subcaption}

\makeatletter
\DeclareRobustCommand{\lintprod}{%
  \mathbin{\mathpalette\int@prod{(0,0)(0.8,0)(0.8,0.6)}}%
}
\DeclareRobustCommand{\rintprod}{%
  \mathbin{\mathpalette\int@prod{(0.1,0.6)(0.1,0)(0.9,0)}}}

\newcommand{\int@prod}[2]{%
  \begingroup
  \sbox\z@{$\m@th#1+$}%
  \setlength\unitlength{\wd\z@}%
  \linethickness{0.09\unitlength}%
  \begin{picture}(1,1)
  \roundcap
  \polyline#2
  \end{picture}%
  \endgroup
}
\makeatother

\counterwithout{section}{chapter}
\settrimmedsize{297mm}{210mm}{*}
\settypeblocksize{664pt}{498.13pt}{*}
\setulmargins{3cm}{*}{*}
\setlrmargins{*}{*}{0.75}
\setmarginnotes{17pt}{51pt}{\onelineskip}
\setheadfoot{\onelineskip}{2\onelineskip}
\setheaderspaces{*}{2\onelineskip}{*}
\checkandfixthelayout

\usepackage[justification=centerlast]{caption}
\usepackage{bm}

\newcommand{\Bb}{{\mathbf B}}
\newcommand{\Eb}{{\mathbf E}}
\newcommand{\Fb}{{\mathbf F}}
\newcommand{\Jb}{{\mathbf J}}

\newcommand{\Rb}{{\mathbf R}}
\newcommand{\Sb}{{\mathbf S}}

\newcommand{\ebf}{{\mathbf e}}

\newcommand{\fb}{{\mathbf f}}

\newcommand{\fbit}{\pmb{\mathit f}}
\newcommand{\jb}{{\mathbf j}}

\newcommand{\xb}{{\mathbf x}}
\newcommand{\yb}{{\mathbf y}}
\newcommand{\vb}{{\mathbf v}}
\newcommand{\ub}{{\mathbf u}}
\newcommand{\wb}{{\mathbf w}}

\newcommand{\drm}{\,\mathrm{d}}

\newcommand{\myparallel}{{\mkern3mu\vphantom{\perp}\vrule depth 0pt\mkern3mu\vrule depth 0pt\mkern3mu}}

\newcommand{\deltabf}{\boldsymbol{\partial}}

\newcommand{\len}[1]{\lvert#1\rvert}

\newcommand{\hodge}{{\scriptscriptstyle\mathcal{H}}}
\newcommand{\hodgeinv}{{\scriptscriptstyle\mathcal{H}^{-1}}}

\makeatletter
\let\save@mathaccent\mathaccent
\newcommand*\if@single[3]{%
  \setbox0\hbox{${\mathaccent"0362{#1}}^H$}%
  \setbox2\hbox{${\mathaccent"0362{\kern0pt#1}}^H$}%
  \ifdim\ht0=\ht2 #3\else #2\fi
  }
\newcommand*\rel@kern[1]{\kern#1\dimexpr\macc@kerna}
\newcommand*\widebar[1]{\@ifnextchar^{{\wide@bar{#1}{0}}}{\wide@bar{#1}{1}}}
\newcommand*\wide@bar[2]{\if@single{#1}{\wide@bar@{#1}{#2}{1}}{\wide@bar@{#1}{#2}{2}}}
\newcommand*\wide@bar@[3]{%
  \begingroup
  \def\mathaccent##1##2{%
    \let\mathaccent\save@mathaccent
    \if#32 \let\macc@nucleus\first@char \fi
    \setbox\z@\hbox{$\macc@style{\macc@nucleus}_{}$}%
    \setbox\tw@\hbox{$\macc@style{\macc@nucleus}{}_{}$}%
    \dimen@\wd\tw@
    \advance\dimen@-\wd\z@
    \divide\dimen@ 3
    \@tempdima\wd\tw@
    \advance\@tempdima-\scriptspace
    \divide\@tempdima 10
    \advance\dimen@-\@tempdima
    \ifdim\dimen@>\z@ \dimen@0pt\fi
    \rel@kern{0.6}\kern-\dimen@
    \if#31
      \overline{\rel@kern{-0.6}\kern\dimen@\macc@nucleus\rel@kern{0.4}\kern\dimen@}%
      \advance\dimen@0.4\dimexpr\macc@kerna
      \let\final@kern#2%
      \ifdim\dimen@<\z@ \let\final@kern1\fi
      \if\final@kern1 \kern-\dimen@\fi
    \else
      \overline{\rel@kern{-0.6}\kern\dimen@#1}%
    \fi
  }%
  \macc@depth\@ne
  \let\math@bgroup\@empty \let\math@egroup\macc@set@skewchar
  \mathsurround\z@ \frozen@everymath{\mathgroup\macc@group\relax}%
  \macc@set@skewchar\relax
  \let\mathaccentV\macc@nested@a
  \if#31
    \macc@nested@a\relax111{#1}%
  \else
    \def\gobble@till@marker##1\endmarker{}%
    \futurelet\first@char\gobble@till@marker#1\endmarker
    \ifcat\noexpand\first@char A\else
      \def\first@char{}%
    \fi
    \macc@nested@a\relax111{\first@char}%
  \fi
  \endgroup
}
\makeatother

\makeatletter
    \def\@endtheorem{\hfill$\P$\endtrivlist\@endpefalse }
\makeatother

\newtheorem{definition}{Definition}
\newtheorem{theorem}{Theorem}
\newtheorem{lemma}{Lemma}

\begin{document}

\title{An Introduction to Space-Time Exterior Calculus\thanks{
This work was funded in part by the Spanish Ministry of Economy and Competitiveness under grants TEC2016-78434-C3-1-R and BES-2017-081360. Published in: \textbf{Mathematics (2019), 7(6), 564. DOI:} \href{https://www.mdpi.com/2227-7390/7/6/564}{10.3390/math7060564}.
}}

\author{\scshape Ivano Colombaro\thanks{ivano.colombaro@upf.edu}, Josep Font-Segura\thanks{josep.font@ieee.org}, Alfonso Martinez\thanks{alfonso.martinez@ieee.org}\thanks{Authors are with the Department of Information and Communication Technologies, Universitat Pompeu Fabra, Barcelona, Spain.}}

\maketitle

\begin{abstract}
The basic concepts of exterior calculus for space-time multivectors are presented: interior and exterior products, interior and exterior derivatives, oriented integrals over hypersurfaces, circulation and flux of multivector fields.
Two Stokes theorems relating the exterior and interior derivatives with circulation and flux respectively are derived. 
As an application, it is shown how the exterior-calculus space-time formulation of the electromagnetic Maxwell equations and Lorentz force recovers the standard vector-calculus formulations, in both differential and integral forms.
\end{abstract}

\section{Introduction} \label{sec-intro}

Vector calculus has, since its introduction by J. W. Gibbs \cite{Gibbs} and Heaviside, been the tool of choice to represent many physical phenomena.
In mechanics, hydrodynamics and electromagnetism, quantities such as forces, velocities and currents are modeled as vector fields in space, while flux, circulation, divergence or curl describe operations on the vector fields themselves.

With relativity theory, it was observed that space and time are not independent but just coordinates in space--time \cite{Minkowski} (pp.~111--120).
Tensors like the Faraday tensor in electromagnetism were quickly adopted as a natural representation of fields in space--time \cite{Ricci-Levi-Civita} (pp.~135--144). In parallel, mathematicians such as Cartan generalized the fundamental theorems of vector calculus, i.e., Gauss, Green, and Stokes, by means of differential forms \cite{Cartan}.
Later on, differential forms were used in Hamiltonian mechanics, e.~g.~to calculate trajectories as vector field integrals \cite{Arnold} (pp.~194--198).

A third extension of vector calculus is given by geometric and Clifford algebras \cite{Clifford}, where vectors are replaced by multivectors and 
operations such as the cross and the dot products subsumed in the geometric product. However, the absence of an explicit formula for the geometric product hinders its widespread use.
An alternative would have been the exterior algebra developed by Grassmann which nevertheless has received little attention in the literature \cite{Grassmann}. An early work in this direction was Sommerfeld's presentation of electromagnetism in terms of six-vectors \cite{Sommerfeld}.

We present a generalization of vector calculus to exterior algebra and calculus. The basic notions of space--time exterior algebra, introduced in Section~\ref{sec-EC-basics},  are extended to exterior calculus in Section~\ref{sec-flux-circ} and applied to rederive the equations of electromagnetism in Section~\ref{sec-maxwell}. In contrast to geometric algebra, our interior and exterior products admit explicit formulations, thereby merging the simplicity and intuitiveness of standard vector calculus with the power of tensors and differential forms. 

\section{Exterior Algebra} \label{sec-EC-basics}

Vector calculus is constructed around the vector space $\Rb ^3$, where every point is represented by three spatial coordinates.
In relativity theory the underlying vector space is $\Rb^{1+3}$ and time is treated as a coordinate in the same footing as the three spatial dimensions. 
We build our theory in space--time with $k$ time dimensions and $n$ space dimensions.
The number of space--time dimensions is thus $k+n$ and we may refer to a $(k,n)$- or $(k+n)$-space--time, $\Rb ^{k+n}$. 
We adopt the convention that the first $k$ indices, i.e., $i=0,\dotsc,k-1$, correspond to time components and the indices $i = k,\dotsc,k+n-1$ represent space components and both $k$ and $n$ are non-negative integers.
A point or position in this space--time is denoted by $\xb$, with components $\smash{\{x_i\}_{i=0}^{k+n-1}}$ in the canonical basis $\smash{\{\ebf_i\}_{i=0}^{k+n-1}}$, that is 
\begin{equation}
\xb = \sum_{i=0}^{k+n-1}x_i\ebf_i	.
\end{equation}

Given two arbitrary canonical basis vectors $\ebf_{i}$ and $\ebf_{j}$, then their dot product in space--time is
\begin{equation}\label{eq:dot_vec_basis}
\ebf_{i}\cdot\ebf_{j} = \begin{cases} -1, & i = j, \, 0\leq i \leq k-1, \\ 
+1, & i = j, \, k\leq i \leq k+n-1, \\
 0, & i \neq j. \end{cases}
\end{equation}
For convenience, 
we define the symbol $\Delta_{ij} = \ebf_{i}\cdot\ebf_{j}$ as the metric diagonal tensor in Minkowski space--time \cite{Minkowski} (pp.~118--120), such that time unit vectors $\ebf_{i}$ have negative norm $\Delta_{ii} = -1$, whereas space unit vectors $\ebf_{i}$ have positive norm $\Delta_{ii} = +1$. The \textbf{dot product} 
of two vectors $\xb$ and $\yb$ is the extension by linearity of the product in~Equation \eqref{eq:dot_vec_basis}, namely
\begin{equation}
 \xb\cdot\yb = \sum_{i=0}^{k+n-1}x_iy_i\Delta_{ii} = -\sum_{i=0}^{k-1}x_iy_i + \sum_{i=k}^{k+n-1}x_iy_i.
\end{equation}


\subsection{Grade, Multivectors, and Exterior Product}

In addition to the $(k+n)$-dimensional vector space $\Rb^{k+n}$ with canonical basis vectors $\ebf_{i}$, there exist other natural vector spaces indexed by ordered lists $I = (i_1,\dotsc,i_m)$ of $m$ non-identical space and time indices for every $m = 0, \dotsc, k+n$. As there are $\binom{k+n}{m}$ such lists, the dimension of this vector space is $\binom{k+n}{m}$. We shall refer to $m$ as \textbf{grade} and to these vectors as \textbf{multivectors} or grade-$m$ vectors if we wish to be more specific. A general multivector can be written as 
\begin{equation}\label{eq:multivector}
\vb = \sum_{I} v_I\ebf_I,
\end{equation}
where the summation extends to all possible ordered lists with $m$ indices. If $m = 0$, the list is empty and the corresponding vector space is $\Rb$. 
The direct sum of these vector spaces for all $m$ is a larger vector space of dimension $\sum_{m=0}^{k+n} \binom{k+n}{m} = 2^{k+n}$, the exterior algebra. In tensor algebra, multivectors correspond to antisymmetric tensors of rank $m$. In this paper, we study vector fields $\vb(\xb)$, namely multivector-valued functions $\vb$ varying over the space--time position $\xb$.

The basis vectors for any grade $m$ may be constructed from the canonical basis vectors $\ebf_i$ by means of the exterior product (also known as wedge product), an operation denoted by $\wedge$ \cite{Winitzki} (p.~2). We identify the vector $\ebf_I$ for the ordered list $I = (i_1,i_2,\dotsc,i_m)$ with the exterior product of $\ebf_{i_1}, \ebf_{i_2} ,\dotsc, \ebf_{i_m}$: 
\begin{equation}
	\ebf_I = \ebf_{i_1}\wedge\ebf_{i_2}\wedge\dotsm\wedge\ebf_{i_m} .
\end{equation}
In general, we may compute the exterior product as follows.
Let two basis vectors $\ebf_I$ and $\ebf_J$ have grades $m = \len{I}$ and $m' = \len{J}$, where $\len{I}$ and $\len{J}$ are the  lengths of the respective index lists. Let $(I,J) = \{i_1,\dotsc,i_m,j_1,\dotsc,j_{m'}\}$ denote the concatenation of $I$ and $J$, let $\sigma(I,J)$ denote the signature of the permutation sorting the elements of this concatenated list of $m+m'$ indices, and let $\varepsilon(I,J)$ denote the resulting sorted list, which we also denote by $I+J$. Then, the exterior product $\ebf_I$ of $\ebf_J$ is defined as
\begin{equation} \label{eq:ext-prod-def}
	\ebf_I\wedge\ebf_J = \sigma(I,J)\ebf_{\varepsilon(I,J)}.
\end{equation}
The \textbf{exterior product} of vectors $\vb$ and $\wb$ is the bilinear extension of the product in~Equation \eqref{eq:ext-prod-def}, 
\begin{equation}
	\vb\wedge\wb = \sum_{I,J}v_I w_J \, \ebf_I\wedge\ebf_J .
\end{equation}
Since permutations with repeated indices have zero signature, the exterior product is zero if \mbox{$m + m' > k+n$} or more generally if both vectors have at least one index in common. Therefore, the exterior product is either zero or a vector of grade $m+m'$.	Further, the exterior product is a skew-commutative operation, as we can also write~Equation \eqref{eq:ext-prod-def} as $\ebf_I\wedge\ebf_J = (-1)^{\len{I}\len{J}}\ebf_J\wedge\ebf_I$.
	
At this point, we define the dot product $\cdot$ for arbitrary grade-$m$ basis vectors $\ebf_I$ and $\ebf_J$ as
\begin{equation}
	\ebf_I\cdot\ebf_J = \Delta_{I,J} = \Delta_{i_1,j_1}\Delta_{i_2,j_2}\dotsm\Delta_{i_m,j_m},
\end{equation}	
where $I$ and $J$ are the ordered lists $I = (i_1,i_2,\dotsc,i_m)$ and $J = (j_1,j_2,\dotsc,j_m)$. As before, we extend this operation to arbitrary grade-$m$ vectors by linearity.

Finally, we define the complement of a multivector. For a unit vector $\ebf_I$ with grade $m$, its Grassmann or Hodge \textbf{complement} \cite{Frankel} (pp.~361--364), denoted by $\ebf_I^\hodge$, is the unit $(k+n-m)$-vector 
\begin{equation} \label{eq:hodge-transf}
	\ebf_I^\hodge = \Delta_{I,I}\sigma(I,I^c)\ebf_{I^c},
\end{equation}
where $I^c$ is the complement of the list $I$, namely the ordered sequence of indices not included in $I$.
As before, $\sigma(I,I^c)$ is the signature of the permutation sorting the elements of the concatenated list $(I,I^c)$ containing all space--time indices.
In other words $\ebf_{I^c}$ is the basis vector of grade $k+n-m$ whose indices are in the complement of $I$. In addition, we define the inverse complement transformation as
\begin{equation} \label{eq:hodge-inv-transf}
	\ebf_I^\hodgeinv = \Delta_{I^c,I^c}\sigma(I^c,I)\ebf_{I^c} .
\end{equation}
We extend the complement and its inverse to general vectors in the space--time algebra by linearity.

\subsection{Interior Products}

While the exterior product of two multivectors is an operation that outputs a multivector whose grade is the addition of the input grades, the dot product takes two multivectors of identical grade and subtracts their grades, yielding a zero-grade multivector, i.e., a scalar. We say that the exterior product raises the grade while the dot product lowers the grade. In this section, we define the left and right interior products of two multivectors as operations that lower the grade and output a multivector whose grade is the difference of the input multivector grades.

As always, we start by defining the operation for the canonical basis vectors. 
Let $\ebf_I$ and $\ebf_J$ be two basis vectors of respective grades $\len{I}$ and $\len{J}$. The left interior product, denoted by $\lintprod$, is defined as
\begin{equation}
	\ebf_I \lintprod \ebf_J = \Delta_{I,I}\sigma\bigl(\varepsilon(I,J^c)^c,I\bigr)\ebf_{\varepsilon(I,J^c)^c}.
	\label{eq:left-int-prod}
\end{equation}
If $I$ is not a subset of $J$, that is when there are elements in $I$ not present in $J$, e.~g.~for $\len{I} > \len{J}$, the signature of the permutation sorting the concatenated list $\bigl(\varepsilon(I,J^c)^c,I\bigr)$ is zero as there are repeated indices in the list to be sorted, and the left interior product is zero. 
Otherwise, if $I$ is a subset of $J$, the permutation rearranges the indices in $J$ in such a way that the last $\len{I}$ positions coincide with $I$ and $\varepsilon(I,J^c)^c$ represents the first $\len{J}-\len{I}$ elements in the rearranged sequence, that is $\varepsilon(I,J^c)^c = J\setminus I$.

The right interior product, denoted by $\rintprod$, of two basis vectors $\ebf_I$ and $\ebf_J$ is defined as
\begin{equation}
	\ebf_I \rintprod \ebf_J = \Delta_{J,J}\sigma\bigl(J,\varepsilon(I^c,J)^c\bigr)\ebf_{\varepsilon(I^c,J)^c}. \label{eq:right-int-prod}
\end{equation}
As with the left interior product, if $J$ is a subset of $I$, $\varepsilon(I^c,J)^c = I\setminus J$ then the permutation rearranges the indices in $I$ so that the first $\len{J}$ positions coincide with $J$, otherwise the right interior product is zero.

In general, we have that $\ebf_I \lintprod \ebf_J = \ebf_J \rintprod \ebf_I (-1)^{\len{I}(\len{J}-\len{I})}$, as verified in Appendix~\ref{sec-app-left-right}. We note that these interior products are not commutative, unless either $\len{J} - \len{I}$ or $\len{I}$ is an even number, 
e.~g.~when $\len{I} = \len{J}$, in which case both interior products coincide with the dot product of the two vectors.  The interior products may therefore be seen as generalizations of the dot product. 

As with the dot and the exterior products, the value of the interior products does not depend on the choice of basis and we may thus compute the \textbf{left interior product} of two vectors $\vb$ and $\wb$ as 
\begin{equation}
	\vb\lintprod\wb = \sum_{I,J}v_I w_J \, \ebf_I\lintprod\ebf_J \,,
\end{equation}
and a similar expression holds for the \textbf{right interior product} $\vb\rintprod\wb$.	Both are grade-lowering operations, as the left (resp.~right) interior product is either zero or a multivector of grade $m'-m$ (resp.~$m-m'$).

The interior products are not independent operations from the exterior product, as they can be expressed in terms of the latter, the Hodge complement and its inverse (proved in Appendix~\ref{sec-app-int-ext}):
\begin{gather}\label{eq:left-int-equiv}
\ebf_I \lintprod \ebf_J = \bigl( \ebf_I \wedge \ebf_J^{\hodge} \bigr)^{\hodgeinv}, \\
\label{eq:right-int-equiv}
\ebf_I \rintprod \ebf_J = \bigl( \ebf_I^{\hodgeinv} \wedge \ebf_J \bigr)^{\hodge}.
\end{gather}

If $\ub$ and $\vb$ are 1-vectors and $\wb$ is an $r$-vector, then we have the following expression
\begin{gather}\label{eq:left-wedge}
	\ub\lintprod(\vb\wedge\wb) = (-1)^{r}(\ub\cdot\vb)\wb + \vb\wedge(\ub\lintprod\wb),
\end{gather}
as proved in Appendix~\ref{sec-app-mix-prod}. This expression can be seen as a generalization of the vectorial expression
\begin{equation}
\mathbf{a} \times (\mathbf{b} \times \mathbf{c}) = (\mathbf{a} \cdot \mathbf{c})\mathbf{b} - (\mathbf{a} \cdot \mathbf{b})\mathbf{c}
\end{equation}
in the vector space $\Rb ^3$, i.e., a $k=0$, $n=3$ space--time. This fact is built of the realization that the cross product between two vectors $\vb$ and $\wb$ can be expressed in the following alternative ways 
\begin{equation}\label{eq:cross-prod-3d}
	\vb\times\wb = (\vb\wedge\wb)^\hodgeinv = \vb\lintprod\wb^\hodgeinv = \vb\lintprod\wb^\hodge.
\end{equation}
Whenever it holds that $I\subseteq J$, the interior and exterior products are related by the following:
\begin{gather}
(\ebf_I\lintprod \ebf_J) \wedge \ebf_I 	=	\Delta_{I,I} \ebf_J,	\\
\ebf_I \wedge (\ebf_J\rintprod \ebf_I) = \Delta_{I,I} \ebf_J.
\end{gather}

Having introduced the basic notions of space--time exterior algebra, the next section focuses on operations with elements in the exterior algebra, namely integrals and derivatives of vector fields.

\section{Integrals and Derivatives of Vector Fields: Circulation and Flux}	\label{sec-flux-circ}

\subsection{Oriented Integrals}

Integrals are, together with derivatives, the fundamental mathematical objects of \textbf{calculus}. For example, operations on vectors fields lying in exterior algebra such as the flux and the circulation are expressed in terms of integrals over high-dimensional geometric objects. The integral of an $m$-graded vector field $\vb$ over a hypersurface ${\mathcal V}^m$ of the same dimension, denoted as
\begin{equation}
\int _{\mathcal{V}^m} \drm ^m\xb\cdot\vb,
\label{eq:integral}
\end{equation}
is the limit of the Riemann sums for the dot product $\drm^{m}\xb\cdot\vb$ over points in the hypersurface, where $\drm^{m}\xb$ is an $m$-dimensional infinitesimal vector element. For any $\ell = 0,\dotsc,k+n$, the infinitesimal vector element $\drm^{\ell}\xb$ is given by the sum of all possible differentials for $\ell$-dimensional hypersurfaces in a $(k,n)$ space--time, and is represented in the canonical basis as
\begin{equation} \label{eq:differential-element}
	\drm^{\ell}\xb = \sum_{I=(i_1,\dotsc,i_\ell)}\drm x_{I}\ebf_I,
\end{equation}
where for a given list $I=(i_1,\dotsc,i_\ell)$ each differential is given by $dx_{I}=\drm x_{i_1}\cdots \drm x_{i_\ell}$.

As in traditional calculus, the integral in~Equation \eqref{eq:integral} exhibits coordinate invariance, while the integrand $\drm ^m\xb\cdot\vb$ is regarded as an oriented object. Orientation is well defined for integrals along a curve from one point to another, or integrals over a surface oriented at the direction of the normal to the surface. Switching the extreme points of the curve, or taking the opposite direction of the normal would induce a change of sign in the line and surface integrals. In our generalization of vector calculus, a positive orientation is implicit in the ordering of the canonical basis. The skew-symmetry property of the exterior product~Equation \eqref{eq:ext-prod-def} may introduce sign changes to compensate an eventual change of orientation after changes of coordinates such as permutations of the space--time components.

For a given hypersurface ${\mathcal V}^m$, a convenient transformation for solving the integral in~Equation \eqref{eq:integral} is one such that, at a given point $\mathbf{x}$ in the hypersurface, the infinitesimal vector element $\drm ^m\xb$ has one component that is \emph{tangent} to the hypersurface at that point. Let $\ebf_\myparallel$ be a unit $m$-graded vector parallel to ${\mathcal V}^m$ at point $\xb$, and let $\ebf_0^\prime,\ldots,\ebf_{k+n-1}^\prime$ form an orthonormal basis of $\Rb^{k+n}$ such that $\ebf_\myparallel=\ebf_{k+n-m}^\prime\wedge \cdots \wedge \ebf_{k+n-1}^\prime$ for the given point $\xb$ in ${\mathcal V}^m$. This change of coordinates from the canonical basis to the new basis is described by a unitary matrix $U$, dependent on $\xb$, and that satisfies
\begin{equation}
	\ebf_0\wedge\dotsm \wedge \ebf_{k+n-1} = \det(U) \, \ebf_0^\prime\wedge\dotsm \wedge \ebf_{k+n-1}^\prime.
	\label{eq:det}
\end{equation}
Being a unitary matrix, the determinant of $U$ is $\pm 1$. Assuming an orientation-preserving change of coordinates, that is $\det(U)=1$,  the infinitesimal vector element in~Equation \eqref{eq:differential-element} for $\ell = m$ can be expressed as
\begin{equation} \label{eq:differential-element2}
	\drm^{m}\xb = \drm x_\myparallel \ebf_\myparallel + \sum_{I=(i_1,\dotsc,i_m)\,:\, I \cap \perp \neq \emptyset}\drm x_{I}\ebf_I^\prime,
\end{equation}
where $\perp=\{0,\ldots,k+n-m-1\}$ is the set of indices for the unit vectors in the new basis orthogonal to ${\mathcal V}^m$. Since all elements in the summation in~Equation \eqref{eq:differential-element2} have at least one differential element lying outside the integration hypersurface, their integrals vanish and therefore
\begin{equation}
	\int _{\mathcal{V}^m} \drm ^m\xb = \int _{\mathcal{V}^m}  dx_\myparallel \ebf_\myparallel.
	\label{eq:parallel}
\end{equation}

In analogy to $\ebf_\myparallel$, a multivector of grade $m$, we define a unit $(k+n-m)$-grade vector $\ebf_\perp$ normal to ${\mathcal V}^m$ at point $\xb$ such that $\ebf_\perp\wedge\ebf_\myparallel=\ebf_0\wedge\dotsm\wedge \ebf_{k+n-1}$. 
From~Equation \eqref{eq:hodge-inv-transf}, we see that one such normal multivector with the correct orientation is
\begin{equation}
	\ebf_\perp = 	\frac{\ebf_\myparallel^\hodgeinv}{\ebf_\myparallel^\hodgeinv\cdot \ebf_\myparallel^\hodgeinv}.
\end{equation}
For the common spaces considered in vector calculus, $\Rb^2$ and $\Rb^3$, and according to~Equation \eqref{eq:det}, orientation-preserving changes of coordinates must respectively satisfy $\ebf_\perp\wedge\ebf_\myparallel= \ebf_0\wedge \ebf_1$ and $\ebf_\perp\wedge\ebf_\myparallel= \ebf_0\wedge \ebf_1 \wedge \ebf_2$, where $\ebf_\perp$ is the basis element normal to ${\mathcal V}^m$. These two equalities turn out to describe the counterclockwise (resp.~right-hand rule) orientation when $\ebf_\perp$ conventionally points outside an integration path for $\Rb^2$ (resp.~a surface for $\Rb^3$) \cite{Arnold} (pp.~184--185).

Building on the concepts and operations of circulation and flux in vector calculus, the right and left interior products lead to general definitions of circulation and flux of multivector fields in exterior algebra along and across hypersurfaces of arbitrary number of dimensions.

\subsection{Circulation and Flux of Multivector Fields}

\begin{definition}
The circulation of a vector field $\vb(\xb)$ of grade $m$ along an $\ell$-dimensional hypersurface $\mathcal{V}^\ell$, denoted by $\mathcal{C}(\vb,\mathcal{V}^\ell)$, is given by
\begin{gather}
	\mathcal{C}(\vb,\mathcal{V}^\ell) = \int_{\mathcal{V}^\ell}\drm^{\ell}\xb\rintprod\vb .
	\label{eq:def1}
\end{gather}
\end{definition}
Expressing the vector field in the canonical basis and using the definition of $\drm^{\ell}\xb$ in~Equation \eqref{eq:differential-element}, the circulation can be specified in some cases of interest. For $\ell=m$, the circulation reads
\begin{equation} \label{eq:circulation-equal}
	\int_{\mathcal{V}^m}\drm^{m}\xb\cdot\vb= \sum_{I=(i_1,\dotsc,i_m)}\Delta_{I,I}\int_{\mathcal{V}^m}dx_{I}v_I.
\end{equation}
For instance, for $\ell = m = 1$ and $\Rb^n$, this formula recovers the definition the circulation of a vector field along a closed path with the appropriate orientation. 

Alternatively, using~Equation \eqref{eq:parallel}, we note that $\vb$ is integrated along the direction of $\ebf_\myparallel$, tangential to the hypersurface, in an orientation-preserving change of coordinates, that is
\begin{equation}
\int _{\mathcal{V}^m} \drm ^m\xb\rintprod\vb = \int _{\mathcal{V}^m}  \drm x_\myparallel  \, \ebf_\myparallel \rintprod \vb.
\label{eq:int1}
\end{equation}
Intuitively, the circulation~Equation \eqref{eq:def1} measures the alignment of an $m$-vector field $\vb$ with respect to ${\mathcal V}^\ell$ for any $\ell$ and $m$, with the circulation being an $(\ell-m)$-vector if $\ell \geq m$ and zero otherwise.


\begin{definition}
The flux of a vector field $\vb(\xb)$ of grade $m$ across an $\ell$-dimensional hypersurface $\mathcal{V}^\ell$, denoted by $\mathcal{F}(\vb,\mathcal{V}^\ell)$, is given by 
\begin{gather}
	\mathcal{F}(\vb,\mathcal{V}^\ell) = \int_{\mathcal{V}^\ell}\drm^{\ell}\xb^{\hodgeinv}\lintprod\vb .
	\label{eq:def2}
\end{gather}
\end{definition}
Expressing both $\vb$ and $\drm^{\ell}\xb$ in the canonical basis, and using the inverse Hodge operation in~Equation \eqref{eq:hodge-inv-transf}, the flux in the special case of $\ell = k + n - m$ can be written as
\begin{equation} \label{eq:flux-equal}
	\int_{\mathcal{V}^\ell}\drm^{\ell}\xb^{\hodgeinv}\cdot\vb = \sum_{I=(i_1,\dotsc,i_{m})}\sigma(I,I^c)\int_{\mathcal{V}^{\ell}}dx_{I^c}v_I.
\end{equation}
As an example in $\Rb^3$, the flux of a vector field $\vb$ through a surface ${\mathcal{V}^2}$ reads
\begin{equation}
	\int _{\mathcal{V}^2} \drm^2\xb^{\hodgeinv}\cdot\vb = \int _{\mathcal{V}^2} \sum_{I,i\notin I} \drm {x_I} \sigma(i,I) \ebf_i \cdot \vb.
	\label{eq:24}
\end{equation}
The right-hand side of~Equation \eqref{eq:24} is a conventional surface integral, upon the identification of $\sum_{I,i\notin I} \drm {x_I} \sigma(i,I) \ebf_i$ as an infinitesimal surface element $\drm\Sb$.

Alternatively, using the analogous of Equation \eqref{eq:parallel} for the differential vector element $\drm ^{\ell}\xb^{\hodgeinv}$, the equivalent to~Equation \eqref{eq:int1} for the flux is
\begin{equation}
\int _{\mathcal{V}^\ell} \drm ^{\ell}\xb^{\hodgeinv}\lintprod\vb = \int _{\mathcal{V}^\ell} \drm x_\myparallel \, \ebf_\myparallel^{\hodgeinv} \lintprod \vb.
\label{eq:int2}
\end{equation}
This equation implies that $\vb$ is integrated along a normal component to the hypersurface since  $\smash{\ebf_\myparallel^{\hodgeinv}}$ is a multivector of grade $k+n-\ell$ orthogonal to ${\mathcal V}^\ell$. 
Intuitively, the flux~Equation \eqref{eq:def2} measures the magnitude of the multivector field crossing the hypersurface. In general, the flux is a vector of grade $(m+\ell-n-k)$ if $\ell \geq k+n-m$ and zero otherwise. For instance, if $\ell = k + n$, the flux of $\vb$ over an $(k+n)$-dimensional hypersurface $\mathcal{V}^{k+n}$ gives the integral of $\vb$ over $\mathcal{V}^{k+n}$, an extension of the volume integral to $\Rb^{k+n}$, 
\begin{equation}
	\int_{\mathcal{V}^{k+n}}\drm^{k+n}\xb^{\hodgeinv}\lintprod\vb = \int_{\mathcal{V}^{k+n}}\drm x_{i_1,\dotsm,i_{k+n}}\vb,
\end{equation}
where we used the relation $1^\hodge = \ebf_{i_1,\dotsc,i_{k+n}}$, implying that $\drm^{k+n}\xb^{\hodgeinv} = \drm x_{i_1,\dotsm,i_{k+n}}$, and that $1\lintprod \vb = \vb$.

\subsection{Exterior and Interior Derivatives}
	
In vector calculus, extensive use is made of the nabla operator $\nabla$, a vector operator that takes partial space derivatives. For instance, operations such as gradient, divergence or curl are expressed in terms of this operator. In our case, we need the generalization to $(k,n)$ space--time to the differential vector operator $\deltabf$, defined as  $(-\partial_0,-\partial_2,\dotsc,-\partial_{k-1}, \partial_{k},\dotsc,\partial_{k+n-1})$, that is
\begin{equation}
	\deltabf = \sum_{i=0}^{k+n-1}\Delta_{ii}\ebf_i\partial_i.
\end{equation}	

For a given vector field $\vb$ of grade $m$, we define the \textbf{exterior derivative} of $\vb$ as $\deltabf\wedge \vb$, namely
\begin{equation}
\deltabf\wedge \vb = \sum_{i=0}^{k+n-1} \sum_I \Delta_{ii}\partial_i v_I \sigma(i,I) \, \ebf_{\varepsilon(i,I)} .
\end{equation}
The grade of the exterior derivative of $\vb$ is $m+1$, unless $m = k+n$, in which case the exterior derivative is zero, as can be deduced from the fact that all signatures are zero.

In addition, we define the \textbf{interior derivative} of $\vb$ as $\deltabf\lintprod \vb$, namely
\begin{equation}
\deltabf\lintprod \vb = 
\sum_{i,I:\, i\in I} \partial_i v_I \sigma(I\setminus i,i) \ebf_{I\setminus i} .
\end{equation}
The grade of the interior derivative of $\vb$ is $m-1$, unless $m = 0$, in which case the interior derivative is zero, as implied by the fact that the grade of $\deltabf$ is larger than the grade of $\vb$. Using~Equation \eqref{eq:left-wedge} with $\ub=\deltabf$ and assuming that  $\vb$ and $\wb$ are 1-vectors, we obtain a generalization of Leibniz's product rule
\begin{gather}
	\deltabf\lintprod(\vb\wedge\wb) =  \vb (\deltabf\cdot\wb) -(\deltabf\cdot\vb)\wb .
\end{gather}

The formulas for the exterior and interior derivatives allow us express some common expressions in vector calculus. For a scalar function $\phi$, its gradient is given by its exterior derivative \mbox{$\nabla\phi = \deltabf\wedge\phi$}, while for a vector field $\vb$, its divergence $\nabla\cdot\vb$ is given by its interior derivative $\nabla\cdot\vb = \deltabf\lintprod\vb$. 
From~Equation \eqref{eq:left-wedge} we further observe that for a scalar function $\phi$ we recover the relation
\begin{equation}
\nabla \cdot  (\nabla \phi) = (\nabla \cdot \nabla) \phi.
\end{equation}
In addition, for a vector fields $\vb$ in $\Rb^3$, taking into account~Equation \eqref{eq:cross-prod-3d} then the curl can be variously expressed as
\begin{equation}\label{eq:link-curl-3d}
	\nabla\times\vb = (\nabla\wedge\vb)^\hodgeinv = \nabla\lintprod\vb^\hodgeinv = \nabla\lintprod\vb^\hodge.
\end{equation}
This formula allows us to write the curl of a vector field $\nabla\times\vb$ in terms of the exterior and interior products and the Hodge complement, while generalizing both the cross product and the curl to grade-$m$ vector fields in space--time algebras with different dimensions. 
Moreover, from~Equation \eqref{eq:left-wedge} we can recover for $r=1$ the well-known formula for the curl of the curl of a vector,
\begin{equation}
\nabla \times ( \nabla \times \vb ) =  \nabla (\nabla \cdot \vb) -\nabla ^2 \vb .
\end{equation}

It is easy to verify that the exterior derivative of an exterior derivative is zero, as is the interior derivative of an interior derivative, that is for any vector field $\vb$, we have that 
\begin{align}
	\deltabf \wedge (\deltabf \wedge \vb) &= 0 \\
	\deltabf \lintprod (\deltabf \lintprod \vb) &= 0.
\end{align}
In regard to the vector space $\Rb ^3$, and using~Equation \eqref{eq:cross-prod-3d}, these expressions imply the well-known facts that the curl of the gradient and the divergence of the curl are zero:
\begin{gather} 
\nabla \times (\nabla \phi) = (\nabla \wedge (\nabla \wedge \phi))^{\hodgeinv} = 0 \\
\nabla \cdot (\nabla \times \vb) = \nabla \lintprod (\nabla \lintprod \vb ^\hodge) = 0.
\end{gather}

\subsection{Stokes Theorem for the Circulation}

In vector calculus in $\Rb^3$, the Kelvin-Stokes theorem for the circulation of a vector field $\vb$ of grade 1 along the boundary $\partial\mathcal{V}^{2}$ of a bidimensional surface $\mathcal{V}^{2}$ relates its value to that of the surface integral of the curl of the vector field over the surface itself. In the notation used in the previous section, the surface integral is the flux of the curl of the vector field across the surface and this theorem reads
\begin{gather}\label{eq:stokes-vector}
	\int_{\partial\mathcal{V}^{2}}\drm\xb\cdot\vb = \int_{\mathcal{V}^{2}}\drm^{2}\xb^{\hodgeinv}\cdot(\nabla\times\vb).
\end{gather}
Taking into account the identity $\nabla\times\vb=(\nabla\wedge\vb)^{\hodgeinv}$ in~Equation \eqref{eq:link-curl-3d}, we  rewrite the right-hand side in~Equation \eqref{eq:stokes-vector} as
\begin{align}\label{eq:stokes-vector-2}
	\int_{\mathcal{V}^{2}}\drm^{2}\xb^{\hodgeinv}\cdot(\nabla\times\vb) &= \int_{\mathcal{V}^{2}}\drm^{2}\xb^{\hodgeinv}\cdot(\nabla\wedge\vb)^{\hodgeinv} \nonumber \\
	&= \int_{\mathcal{V}^{2}}\drm^{2}\xb\cdot(\nabla\wedge\vb),
\end{align}
where we used that $\ub\cdot\wb = \ub^{\hodgeinv}\cdot\wb^{\hodgeinv} = \ub^{\hodge}\cdot\wb^{\hodge}$ for vectors $\ub$, $\wb$. The \textit{flux} of the curl of the vector field \textit{across} a surface is also the \textit{circulation} of the exterior derivative of the vector field \textit{along} that surface.
 
The generalized Stokes theorem for differential forms \cite{Cartan} (p.~80) allows us to extend the Kelvin-Stokes theorem to multivectors of any grade $m$ as we do in the following theorem. 
\begin{theorem} \label{theo:circ}
The circulation of a grade-$m$ vector field $\vb$ along the boundary $\partial\mathcal{V}^{\ell}$ of an $\ell$-dimensional hypersurface $\mathcal{V}^{\ell}$ is equal to the circulation of the exterior derivative of $\vb$ along $\mathcal{V}^{\ell}$:
\begin{equation} \label{eq:circulation}
	\mathcal{C}(\vb,\partial\mathcal{V}^{\ell}) = \mathcal{C}(\deltabf\wedge \vb,\mathcal{V}^{\ell}).
\end{equation}
\end{theorem}
As hinted at above, the role of the vector curl in the right-hand side of~Equation \eqref{eq:stokes-vector} is played by the exterior derivative in this generalized theorem.

\begin{proof}
We start by stating the generalized Stokes Theorem for differential forms \cite{Cartan} (pp.~80)
\begin{equation}\label{eq:stokes-diff-forms}
 \int_{\partial\mathcal{V}^{\ell}} \omega = \int_{\mathcal{V}^{\ell}}\drm\omega ,
\end{equation}
where $\omega$ is a differential form and $\drm\omega$ its exterior derivative, represented by the operator
\begin{equation}\label{eq:ext-der-omega}
\drm = \sum _j \drm x_j \partial _j.
\end{equation} 

Expressing the circulations in~Equation \eqref{eq:circulation} by means of the integrals in~Equation \eqref{eq:def1}, we obtain
\begin{equation}\label{eq:circulation-2}
	\int_{\partial\mathcal{V}^{\ell}}\drm^{\ell-1}\xb\rintprod\vb = \int_{\mathcal{V}^{\ell}}\drm^{\ell}\xb\rintprod(\deltabf\wedge \vb). 
\end{equation}
In the integral in the left-hand side of~Equation \eqref{eq:circulation-2}, the integrand is a differential form $\omega = \drm^{\ell-1}\xb\rintprod\vb$. After expanding the interior product using the definitions of $\drm^{\ell-1}\xb$ and $\vb$ we obtain
\begin{align}
\omega = \Biggl( \sum _{J_{\ell-1}} \drm x_J \, \ebf_J \Biggr) \rintprod \Biggl( \sum _{I_m} v_I \, \ebf_I \Biggr)
= \sum _{J_{\ell-1}, I_m \,:\, I\subseteq J} \Delta_{I,I} \sigma(I, \varepsilon(I,J^c)^c) v_I \drm x_J \, \ebf_{\varepsilon(I,J^c)^c}.
\end{align}
Then, computing the exterior derivative of this form with~Equation \eqref{eq:ext-der-omega} gives
\begin{equation}\label{eq:ext_der_stokes}
\drm \omega =  \sum _{J_{\ell-1}, I_m \,:\, I\subseteq J} \sum _{j \not\in J} \Delta_{I,I} \sigma(I, \varepsilon(I,J^c)^c) \partial _j v_I \sigma(j,J) \drm x_{\varepsilon(j,J)} \, \ebf_{\varepsilon(I,J^c)^c} .
\end{equation}

We next write down the integrand in the right-hand side of~Equation \eqref{eq:circulation-2}, $\drm^{\ell}\xb\rintprod(\deltabf\wedge \vb)$, that is
\begin{align}
\drm^{\ell}\xb\rintprod(\deltabf\wedge \vb) &= \Biggl( \sum _{K_{\ell +1}} \drm x_K \, \ebf_K \Biggr) \rintprod \Biggl( \sum _{I_m} \sum_{j\not\in I} \Delta_{j,j} \partial_j v_I \, \sigma(j,I) \, \ebf_{\varepsilon(j,I)} \Biggr)  \nonumber \\
&= \sum _{K_{\ell},I_m \,:\, \varepsilon(j,I) \subseteq K_{\ell}} \sum_{j\not\in I} \Delta_{j,j} \partial_j v_I \drm x_K \, \sigma(j,I) \Delta_{\varepsilon(j,I),\varepsilon(j,I)} \, \sigma(\varepsilon(j,I), \varepsilon(K^c,\varepsilon(j,I))^c) \, \ebf_{\varepsilon(K^c,\varepsilon(j,I))^c} 	\nonumber \\
&= \sum _{K_{\ell},I_m \,:\, \varepsilon(j,I) \subseteq K_{\ell+1}} \sum_{j\not\in I} \Delta_{I,I} \partial_j v_I \drm x_K \, \sigma(j,I) \sigma(\varepsilon(j,I), \varepsilon(K^c,\varepsilon(j,I))^c) \, \ebf_{\varepsilon(K^c,\varepsilon(j,I))^c} ,\label{eq:31}
\end{align}
and verify that it coincides with exterior derivative in~Equation \eqref{eq:ext_der_stokes}. As the set of $m$ indices $I_m$ is included in the sets $J_{\ell-1}$ or $K_{\ell}$ in~Equation \eqref{eq:ext_der_stokes} or~\eqref{eq:31}, we may write $K_{\ell} = \varepsilon(J_{\ell-1},j)$ for some $j\notin J_{\ell-1}$. Then, we obtain the following chain of equalities for the basis elements in~Equations \eqref{eq:ext_der_stokes} and \eqref{eq:31}:
%
%
\begin{equation}
\ebf_{\varepsilon(K^c,\varepsilon(j,I))^c} = \ebf_{\varepsilon(K^c \cup \{j\} \cup I)^c} = \ebf_{\varepsilon(J^c \setminus \{j\} \cup \{j\} \cup I)^c} = \ebf_{\varepsilon(J^c \cup I)^c} = \ebf_{\varepsilon(I,J^c)^c}.
\end{equation} 
Therefore, and using that $\varepsilon(J^c \setminus \{j\},\varepsilon(j,I))^c = J\setminus I$, we can write~Equation \eqref{eq:31} as
\begin{align}
\drm^{\ell}\xb\rintprod(\deltabf\wedge \vb) &= \sum _{J_{\ell-1},I_m,j \not\in I \,:\, \varepsilon(j,I) \subseteq J\cup \{j\}} \Delta_{I,I} \partial_j v_I \drm x_{\varepsilon(J,j)} \, \sigma(j,I) \sigma(\varepsilon(j,I), J\setminus I) \, \ebf_{\varepsilon(I,J^c)^c}.\label{eq:35}
\end{align}

Comparing~Equation \eqref{eq:35} with~Equation \eqref{eq:ext_der_stokes}, the expressions coincide if this identity holds:
\begin{equation}
\label{eq:sigma-circ}
\sigma(j,J) \sigma(I, J  \setminus I) = \sigma(j,I) \sigma(j+ I, J \setminus I) .
\end{equation}
To prove~Equation \eqref{eq:sigma-circ} we exploit that the $\sigma$ are permutation signatures and that the signature of the composition of permutations is the product of the respective signatures. We proceed with the help of a visual aid in \figurename~\ref{fig:circ}, which depicts the identity between two different ways of sorting the concatenated list $(j,I,J\setminus I)$. On the left column we first sort the list $(I,J\setminus I)$ to obtain $J$ and then sort the list $(j,J)$. On the right column, we first sort the list $(j,I)$ and then the list $(j+I,J\setminus I)$.
This proves~Equation \eqref{eq:sigma-circ} and the theorem.
\begin{figure}[htb]
\centering
	\begin{subfigure}[t]{0.45\textwidth}
	\centering
\begin{tikzpicture}
\draw [blue, thick] (-3,4) node[black] {$\vert$} -- (1,4) node[black] {$\vert$};
\node[blue, above] at (-1,4) {$I$};
\draw [thick] (-3.5,4) node[black] {$\vert$} -- (-3,4) node[black] {$\vert$};
\node[above] at (-3.25,4) {$j$};
\draw [red, thick] (1,4) -- (3,4) node[black] {$\vert$};
\node[red, above] at (2,4) {$J\setminus I$};
\draw [blue, thick] (-3,3)node[black] {$\vert$} -- (3,3) node[black] {$\vert$};
\node[blue, above] at (0,3) {$J$};
\draw [thick] (-3.5,3) node[black] {$\vert$} -- (-3,3) node[black] {$\vert$};
\node[above] at (-3.25,3) {$j$};
\draw [thick] (-3.5,2)node[black] {$\vert$} -- (3,2) node[black] {$\vert$};
\node[above] at (-0.25,2) {$\varepsilon(j,J)$};
\end{tikzpicture}
\end{subfigure}
\centering
	\begin{subfigure}[t]{0.45\textwidth}
	\centering
\begin{tikzpicture}
\draw [blue, thick] (-3,4) node[black] {$\vert$} -- (1,4) node[black] {$\vert$};
\node[blue, above] at (-1,4) {$I$};
\draw [thick] (-3.5,4) node[black] {$\vert$} -- (-3,4) node[black] {$\vert$};
\node[above] at (-3.25,4) {$j$};
\draw [red, thick] (1,4) -- (3,4) node[black] {$\vert$};
\node[red, above] at (2,4) {$J\setminus I$};
\draw [blue, thick] (-3.5,3) node[black] {$\vert$} -- (1,3) node[black] {$\vert$};
\node[blue, above] at (-1.25,3) {$j+I$};
\draw [red, thick] (1,3) -- (3,3) node[black] {$\vert$};
\node[red, above] at (2,3) {$J\setminus I$};
\draw [thick] (-3.5,2)node[black] {$\vert$} -- (3,2) node[black] {$\vert$};
\node[above] at (-0.25,2) {$\varepsilon(j,J)$};
\end{tikzpicture}
\end{subfigure}
\caption{Visual aid for the identity among permutations in~Equation \eqref{eq:sigma-circ}.}
\label{fig:circ}
\end{figure}
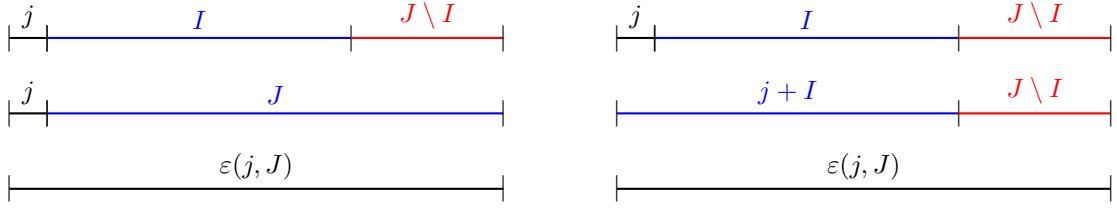

Finally, we note that, had we defined the circulation with the left interior product, we would have got an incompatible relation in~Equation \eqref{eq:sigma-circ}, which could not be solved. 
\end{proof}

\subsection{Stokes Theorem for the Flux}	

In vector calculus in $\Rb^3$, the Gauss theorem relates the volume integral of the divergence of a vector field $\vb$ over a region $\mathcal{V}^{3}$ to the surface integral of the vector field over the region boundary $\partial\mathcal{V}^{3}$. In the notation used in previous sections, and taking into account that both the surface integral and the volume integral can be expressed as fluxes for $\Rb^3$, this theorem reads
\begin{gather}\label{eq:gauss-vector}
	\int_{\partial\mathcal{V}^{3}}\drm^2\xb^{\hodgeinv}\cdot\vb = \int_{\mathcal{V}^{3}}\drm^{3}\xb^{\hodgeinv}(\nabla\cdot\vb).
\end{gather}
Making use of the identity $\nabla\cdot\vb=\nabla\lintprod\vb$, we can rewrite the right-hand side in~Equation \eqref{eq:gauss-vector} as
\begin{align}
	\int_{\mathcal{V}^{3}}\drm^{3}\xb^{\hodgeinv}(\nabla\cdot\vb) = \int_{\mathcal{V}^{3}}\drm^{3}\xb^{\hodgeinv}(\nabla\lintprod\vb).
\end{align}
In other words, the Gauss theorem relates the flux of the interior derivative of a vector field $\vb$ across a region $\mathcal{V}^{3}$ to the flux of the vector field itself across the region boundary $\partial\mathcal{V}^{3}$. 

The generalized Stokes theorem for differential forms allows us to extend the Gauss theorem to multivectors of any grade $m$ as we do in the following theorem. 
\begin{theorem} \label{theo:flux}
The flux of a grade-$m$ vector field $\vb$ across the boundary $\partial\mathcal{V}^{\ell}$ of an $\ell$-dimensional hypersurface $\mathcal{V}^{\ell}$ is equal to the flux of the interior derivative of $\vb$ across $\mathcal{V}^{\ell}$:
\begin{equation} \label{eq:flux}
	\mathcal{F}(\vb,\partial\mathcal{V}^{\ell}) = \mathcal{F}(\deltabf\lintprod \vb,\mathcal{V}^{\ell}).
\end{equation}
\end{theorem}

\begin{proof}
Expressing the fluxes in~Equation \eqref{eq:flux} by means of the integrals in~Equation \eqref{eq:def2}, we obtain
\begin{equation}\label{eq:flux-2}
	\int_{\partial\mathcal{V}^{\ell}}\drm^{\ell-1}\xb^{\hodgeinv}\lintprod\vb = 
	\int_{\mathcal{V}^{\ell}}\drm^{\ell}\xb^{\hodgeinv}\lintprod(\deltabf\lintprod \vb).
\end{equation}

As in the proof of Theorem~\ref{theo:circ}, we apply the Stokes theorem for differential forms in~Equation \eqref{eq:stokes-diff-forms} upon the identifications  $\omega$ with $\drm^{\ell-1}\xb^{\hodgeinv}\lintprod\vb$ and $\drm \omega$ with $\drm^{\ell}\xb^{\hodgeinv}\lintprod(\deltabf\lintprod \vb)$. First, for $\omega$, we get
\begin{align}
\Biggl( \sum_{J_{\ell-1}}\drm x_J \Delta_{J^c, J^c} \sigma(J^c,J) \ebf_{J^c}  \Biggr) \lintprod \Biggl( \sum _{I_m} v_I \, \ebf_I \Biggr) &= \sum _{J_{\ell-1}, I_m \,:\, J^c \subseteq I} v_I \drm x_J \sigma (J^c, J)\sigma(\varepsilon(J^c, I^c)^c, J^c) \, \ebf_\varepsilon(J^c, I^c)^c	\nonumber \\
&= \sum _{J_{\ell-1}, I_m \, : \, J^c \subseteq I} v_I \drm x_J \sigma (J^c, J) \sigma(I\setminus J^c, J^c) \, \ebf_{I\setminus J^c}.
\label{eq:55}
\end{align}
Now, taking the exterior derivative of~Equation \eqref{eq:55}, we obtain
\begin{equation}\label{eq:56}
\drm \omega = \sum _{J_{\ell-1}, I_m \,:\, J^c \subseteq I } \sum_{j \not \in J} \partial_j v_I \drm x_{\varepsilon(j,J)} \sigma(j,J) \sigma (J^c, J) \sigma(I\setminus J^c, J^c) \, \ebf_{I\setminus J^c} .
\end{equation}
This quantity should be equal to $\drm^{\ell}\xb^{\hodgeinv}\lintprod(\deltabf\lintprod \vb)$ in the right-hand side of~Equation \eqref{eq:flux-2}, which we expand as
\begin{align}
\drm^{\ell}\xb^{\hodgeinv}\lintprod(\deltabf\lintprod \vb) &= \Biggl( \sum_{K_{\ell}} \drm x_K\Delta_{K^c, K^c}\sigma(K^c, K) \ebf_{K^c} \Biggr) \lintprod \Biggl( \sum_{I : j\in I}\partial_j v_I \sigma(I\setminus j,j) \ebf_{I\setminus j} \Biggr)	\nonumber \\
&=\sum _{K_{\ell}, I_m \,:\, K^c \subseteq I\setminus j} \sum_{j\in I} \partial_j v_I\drm x_K \sigma(K^c, K) \sigma(I\setminus j,j) \sigma(I\setminus j \setminus K^c, K^c) \, \ebf_{I\setminus j \setminus K^c} .\label{eq:58}
\end{align}

We first consider the sets in the summations in the alternative expressions for $\drm \omega$,~Equations \eqref{eq:56} and \eqref{eq:58}. Since $J^c$ contains $j$ and is a subset of $I$, but $K^c$ does not contain $j$ and is also a subset of $I$ (with $j\in I$), then we can assert that  $K = J \cup \{j\}$ so that the conditions in the summations are equivalent. The basis elements coincide and so do the differentials and derivatives, and it remains to verify the identity 
\begin{equation}
\sigma(j,J) \sigma(J^c, J)\sigma(I\setminus J^c, J^c) = \sigma(K^c, K) \sigma(I\setminus j,j) \sigma(I\setminus j \setminus K^c, K^c).
\end{equation}
With the definition $L = I\setminus J^c$, and expressed in terms of $j$, $J$, and $L$, this condition gives
\begin{equation}
\sigma(j,J) \sigma(J^c, J)\sigma(L, J^c) = \sigma(J^c \setminus j, J+ j) \sigma(J^c\setminus j+L,j)  \sigma(L, J^c \setminus j) .
\end{equation}
Multiplying both sides of the equation by $\sigma(J^c,J)$, $\sigma(J^c \setminus j, J+ j)$ and $\sigma(J^c\setminus j,j)$, and taking into account that the square of a signature is $+1$, we obtain
\begin{equation} \label{eq:perm-flux-proof}
\sigma(J^c\setminus j,j)\sigma(L, J^c)\sigma(j, J) \sigma(J^c \setminus j, J+ j) =   \sigma(L, J^c \setminus j) \sigma(J^c\setminus j+L,j) \sigma(J^c\setminus j,j) \sigma(J^c,J).
\end{equation}

We start by simplifying~Equation \eqref{eq:perm-flux-proof} by noting that
\begin{equation}
\sigma(J^c\setminus j,j)\sigma(L, J^c) = \sigma(L, J^c \setminus j) \sigma(J^c\setminus j+L,j),
\end{equation}
with help of the visual aid in~\figurename~\ref{fig:flux1}. The permutations on the left column first merge $(J^c \setminus j)$ with $j$ and then the resulting $J^c$ with L. Similarly, on the right column, we start with $L$, $(J^c \setminus j)$ and $\{j\}$, then concatenate $(L, J^c \setminus j) $ and then add $j$, getting the same result as the left column.
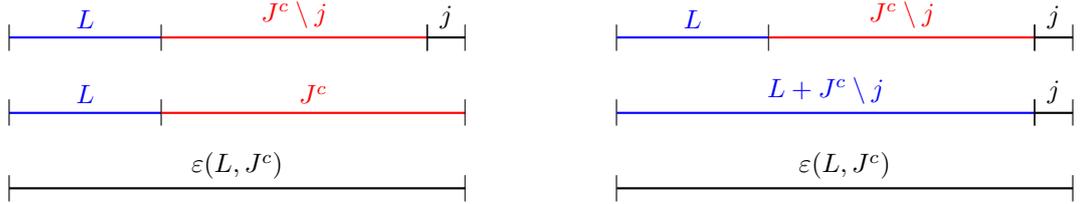
\begin{figure}[htb]
\centering
	\begin{subfigure}[t]{0.45\textwidth}
	\centering
\begin{tikzpicture}
\draw [blue, thick] (-3,4) node[black] {$\vert$} -- (-1,4) node[black] {$\vert$};
\node[blue, above] at (-2,4) {$L$};
\draw [red, thick] (-1,4) -- (2.5,4) node[black] {$\vert$};
\node[red, above] at (0.75,4) {$J^c\setminus j$};
\draw [thick] (2.5,4) node[black] {$\vert$} -- (3,4) node[black] {$\vert$};
\node[above] at (2.75,4) {$j$};
\draw [blue, thick] (-3,3) node[black] {$\vert$} -- (-1,3) node[black] {$\vert$};
\node[blue, above] at (-2,3) {$L$};
\draw [red, thick] (-1,3) -- (3,3) node[black] {$\vert$};
\node[red, above] at (1,3) {$J^c$};
\draw [thick] (-3,2)node[black] {$\vert$} -- (3,2) node[black] {$\vert$};
\node[above] at (0,2) {$\varepsilon(L,J^c)$};
\end{tikzpicture}
\end{subfigure}
\centering
	\begin{subfigure}[t]{0.45\textwidth}
	\centering
\begin{tikzpicture}
\draw [blue, thick] (-3,4) node[black] {$\vert$} -- (-1,4) node[black] {$\vert$};
\node[blue, above] at (-2,4) {$L$};
\draw [red, thick] (-1,4) -- (2.5,4) node[black] {$\vert$};
\node[red, above] at (0.75,4) {$J^c\setminus j$};
\draw [thick] (2.5,4) node[black] {$\vert$} -- (3,4) node[black] {$\vert$};
\node[above] at (2.75,4) {$j$};
\draw [blue, thick] (-3,3) node[black] {$\vert$} -- (2.5,3) node[black] {$\vert$};
\node[blue, above] at (-0.25,3) {$L+J^c\setminus j$};
\draw [thick] (2.5,3) node[black] {$\vert$} -- (3,3) node[black] {$\vert$};
\node[above] at (2.75,3) {$j$};
\draw [thick] (-3,2)node[black] {$\vert$} -- (3,2) node[black] {$\vert$};
\node[above] at (0,2) {$\varepsilon(L,J^c)$};
\end{tikzpicture}
\end{subfigure}
\caption{Visual aid for the identity $\sigma(J^c\setminus j,j)\sigma(L, J^c) = \sigma(L, J^c \setminus j) \sigma(J^c\setminus j+L,j)$.}
\label{fig:flux1}
\end{figure}

Therefore, we have reduced~Equation \eqref{eq:perm-flux-proof} to the simpler form
\begin{equation}
\sigma(j, J) \sigma(J^c \setminus j, J+ j) =    \sigma(J^c\setminus j,j) \sigma(J^c,J),
\end{equation}
which we prove with the aid depicted in~\figurename~\ref{fig:flux2}. On the left column, $j$ and $J$ are first merged and then the concatenation $(J^c \setminus j, J+ j)$ gives the sorted $\varepsilon(J^c, J)$. On the right column, after sorting $(J^c\setminus j)$ with $j$, merging it with $J$ leads to the same final sequence.
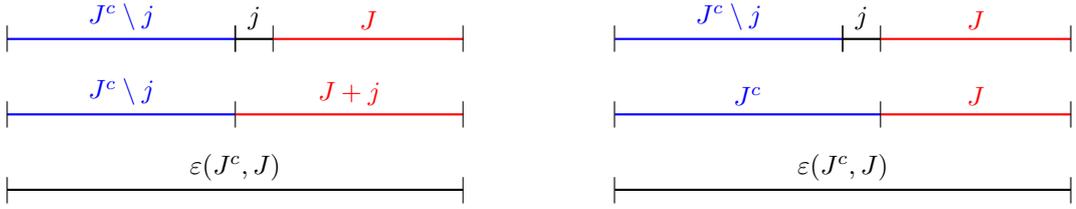
\begin{figure}[htb]
\centering
	\begin{subfigure}[t]{0.45\textwidth}
	\centering
\begin{tikzpicture}
\draw [blue, thick] (-3,4) node[black] {$\vert$} -- (0,4) node[black] {$\vert$};
\node[blue, above] at (-1.5,4) {$J^c\setminus j$};
\draw [thick] (0,4) node[black] {$\vert$} -- (0.5,4) node[black] {$\vert$};
\node[above] at (0.25,4) {$j$};
\draw [red, thick] (0.5,4) -- (3,4) node[black] {$\vert$};
\node[red, above] at (1.75,4) {$J$};
\draw [blue, thick] (-3,3) node[black] {$\vert$} -- (0,3) node[black] {$\vert$};
\node[blue, above] at (-1.5,3) {$J^c\setminus j$};
\draw [red, thick] (0,3) -- (3,3) node[black] {$\vert$};
\node[red, above] at (1.5,3) {$J+j$};
\draw [thick] (-3,2)node[black] {$\vert$} -- (3,2) node[black] {$\vert$};
\node[above] at (0,2) {$\varepsilon(J^c,J)$};
\end{tikzpicture}
\end{subfigure}
\centering
	\begin{subfigure}[t]{0.45\textwidth}
	\centering
\begin{tikzpicture}
\draw [blue, thick] (-3,4) node[black] {$\vert$} -- (0,4) node[black] {$\vert$};
\node[blue, above] at (-1.5,4) {$J^c\setminus j$};
\draw [thick] (0,4) node[black] {$\vert$} -- (0.5,4) node[black] {$\vert$};
\node[above] at (0.25,4) {$j$};
\draw [red, thick] (0.5,4) -- (3,4) node[black] {$\vert$};
\node[red, above] at (1.75,4) {$J$};
\draw [blue, thick] (-3,3) node[black] {$\vert$} -- (0.5,3) node[black] {$\vert$};
\node[blue, above] at (-1.25,3) {$J^c$};
\draw [red, thick] (0.5,3) -- (3,3) node[black] {$\vert$};
\node[red, above] at (1.75,3) {$J$};
\draw [thick] (-3,2)node[black] {$\vert$} -- (3,2) node[black] {$\vert$};
\node[above] at (0,2) {$\varepsilon(J^c,J)$};
\end{tikzpicture}
\end{subfigure}
\caption{Visual aid for the identity $\sigma(j, J) \sigma(J^c \setminus j, J+ j) =    \sigma(J^c\setminus j,j) \sigma(J^c,J)$.}
\label{fig:flux2}
\end{figure}
\end{proof}


\section{An Application to Electromagnetism in 1+3 Dimensions} \label{sec-maxwell}
In this section, we show how to recover the standard form of Maxwell equations and Lorentz force in $1+3$ dimensions from a formulation with \textit{exterior calculus} involving an electromagnetic bivector field $\Fb$ and a 4-dimensional current density vector $\Jb$. In the appropriate units, the bivector field
$\Fb$ can be decomposed as $\Fb = \Fb _\Eb + \Fb _\Bb $, where $\Fb _\Eb$
contains the electric-field $\Eb$ time-space components and $\Fb _\Bb$ contains the space-space components for the magnetic field $\Bb$. Similarly, the current density depends on the charge density $\rho$ and the spatial current density $\jb$. More specifically, 
\begin{gather} 
\Jb=\rho\ebf_0 + \jb \\
\Fb = \Fb _\Eb + \Fb _\Bb = \ebf_0 \wedge \Eb + \Bb ^\hodge. \label{eq:split-F}
\end{gather}
Here the Hodge complement acts only on the space components, and $\Bb ^\hodge = \Bb ^\hodgeinv$. The bivector field $\Fb$ is closely related to the Faraday tensor, a rank-2 antisymmetric tensor.

Maxwell equations, in their differential form, constrain the divergence of the electric and the magnetic field,~Equations \eqref{eq:diff-max-1} and~\eqref{eq:diff-max-2}, respectively, and the curl of $\Eb$ and $\Bb$, namely~Equations \eqref{eq:diff-max-3} and \eqref{eq:diff-max-4} \cite{Feynman} (p.~4-1).
\begin{gather}
\nabla \cdot \Eb = \rho  \label{eq:diff-max-1}\\
\nabla \cdot \Bb = 0 	\label{eq:diff-max-2} \\
\nabla\times\Eb = -\partial_0 \Bb 	\label{eq:diff-max-3}	\\
\nabla \times \Bb = \partial_0 \Eb  + \jb.  \label{eq:diff-max-4}
\end{gather}
We refer to~Equations \eqref{eq:diff-max-2} and \eqref{eq:diff-max-3} as homogeneous Maxwell equations and to~Equations \eqref{eq:diff-max-1} and \eqref{eq:diff-max-4} as inhomogeneous Maxwell equations, as they include the fields and the sources given by charge and current densities. In exterior-calculus notation, both pairs of equations can be combined into simple multivector equations, 
\begin{align}
\deltabf \wedge \Fb &= 0 \label{eq:max-diff-homo} \\
\deltabf \lintprod \Fb &= \Jb, \label{eq:max-diff-nonhomo}
\end{align}
where $\deltabf$ is the differential operator $\deltabf =  -\partial_0 \ebf_0 + \nabla$ for $k=1$ and $n=3$. As a consistency check, note that the wedge product raises the grade of $\Fb$, and the zero in~Equation \eqref{eq:max-diff-homo} is the zero trivector; also, as the left interior product lowers the grade of $\Fb$, both sides of~Equation \eqref{eq:max-diff-nonhomo} relate space--time vectors.

Next to Maxwell equations, the Lorentz force density $\fbit$ characterizes, after integrating over the appropriate region, the force exerted by the electromagnetic field upon a system of charges described by the charge and current densities $\rho$ and $\jb$ \cite{Feynman} (pp.~13-1--13-3),
\begin{equation}
\fbit = \rho \Eb  + \jb\times\Bb.
\end{equation}
In relativistic form, the Lorentz force density becomes a four-dimensional vector $\fb$ \cite{Minkowski} (pp.~153--157). 
The time component of this vector is $\jb\cdot \Eb$, the power dissipated per unit of volume, or after integrating over the appropriate region,
the rate of work being done on the charges by the fields. 
In exterior-calculus notation, the Lorentz force density vector can be computed as a left interior product, namely
\begin{equation}\label{eq:lorentz-force}
 \fb = \Jb \lintprod \Fb.
 \end{equation}
 

\subsection{Equivalence of the Lorentz Force Density}

In this section, we prove that~Equation \eqref{eq:lorentz-force} indeed recovers the relativistic Lorentz force density by verifying that its components in both vector-calculus and exterior-calculus coincide.
From the definitions of $\Jb$ and $\Fb$, and using the distributive property of the interior product, we get
\begin{align}
 \fb &= (\rho\ebf_0 + \jb) \lintprod (\Fb _\Eb + \Fb _\Bb) \nonumber \\
 &= \rho\ebf_0\lintprod \Fb _\Eb + \rho\ebf_0\lintprod \Fb _\Bb + \jb\lintprod \Fb _\Eb + \jb\lintprod\Fb _\Bb \nonumber  \\
 &= \rho\ebf_0\lintprod (\ebf_0 \wedge \Eb) + \rho\ebf_0\lintprod \Bb ^\hodge + \jb\lintprod (\ebf_0 \wedge \Eb) + \jb\lintprod\Bb ^\hodgeinv. 
\end{align}
Some straightforward calculations give $\ebf_0\lintprod (\ebf_0 \wedge \Eb) = \Eb$, $\ebf_0\lintprod \Bb ^\hodge = 0$, and $\jb\lintprod (\ebf_0 \wedge \Eb) = \ebf_0\jb\cdot \Eb$. In addition, the formula for the left interior product in~Equation \eqref{eq:cross-prod-3d} gives $\jb\lintprod\Bb ^\hodgeinv = (\jb\wedge\Bb)^\hodgeinv = \jb\times\Bb$, where the cross product is only valid for three dimensions. With these calculations, we obtain
\begin{align}
 \fb &= \rho \Eb  + \ebf_0\jb\cdot \Eb + \jb\times\Bb,
\end{align}
namely, a time-component $\jb\cdot \Eb$ and a spatial component equal to the Lorentz force density $ \rho \Eb  + \jb\times\Bb$. 

\subsection{Equivalence of the Differential Form of Maxwell Equations}

In this section, we prove that~Equation \eqref{eq:max-diff-homo} indeed recovers the homogeneous Maxwell equations and that~Equation \eqref{eq:max-diff-nonhomo} recovers the inhomogeneous Maxwell equations.

First, we observe that the exterior derivative $\deltabf \wedge \Fb$ gives a trivector with 4 components, while the homogeneous Maxwell equations are a scalar, Equation \eqref{eq:diff-max-2}, and a vector, Equation \eqref{eq:diff-max-3}. We shall verify that the scalar equation turns out to be given by the trivector component $\ebf_{123}$ of $\deltabf \wedge \Fb$, while the vector equation is given by the trivector components $\ebf_{012}$, $\ebf_{013}$, and $\ebf_{023}$ of the exterior derivative.

We evaluate the exterior derivative $\deltabf \wedge \Fb$ using the decomposition of $\Fb$ in~Equation \eqref{eq:split-F}, 
\begin{align}
\deltabf \wedge \Fb &= -\partial_0 \ebf_0 \wedge\ebf_0 \wedge \Eb  -\partial_0 \ebf_0 \wedge \Bb ^\hodge + \nabla\wedge\ebf_0\wedge\Eb + \nabla \wedge \Bb^\hodge \nonumber \\
&= -\partial_0 \ebf_0 \wedge \Bb ^\hodge - \ebf_0\wedge(\nabla\wedge\Eb) + \nabla \wedge \Bb^\hodge\label{eq:54}	\nonumber \\
&= -\ebf_0 \wedge (\partial_0 \Bb ^\hodge + \nabla\wedge\Eb) + \nabla \wedge \Bb^\hodge,
\end{align}
where we used that $\ebf_0\wedge\ebf_0 = 0$ and that $\nabla\wedge\ebf_0 = -\ebf_0\wedge\nabla$ in the second step of~Equation \eqref{eq:54}.
Taking advantage of~Equation \eqref{eq:link-curl-3d} we have the equality $\nabla\wedge\Eb = (\nabla\times\Eb)^\hodge$, while $\nabla \wedge \Bb^\hodge = (\nabla \cdot \Bb )^\hodge$, and
\begin{align}
\deltabf \wedge \Fb &= -\ebf_0 \wedge (\partial_0 \Bb ^\hodge + (\nabla\times\Eb)^\hodge) + (\nabla \cdot \Bb )^\hodge.
\end{align}

Indeed, the first summand vanishes when $\partial_0 \Bb ^\hodge + (\nabla\times\Eb)^\hodge = 0$ or, taking the inverse Hodge complement, when~Equation \eqref{eq:diff-max-3} holds. In terms of components, the spatial Hodge complement in this equation transforms a spatial vector into a bivector with components $\ebf_{12}$, $\ebf_{13}$, and $\ebf_{23}$ only and this equation recovers the homogeneous Maxwell equation in~Equation \eqref{eq:diff-max-3}. After taking the exterior product with $\ebf_0$, we obtain the trivector components $\ebf_{012}$, $\ebf_{013}$, and $\ebf_{023}$.
Similarly, the second term vanishes for $(\nabla \cdot \Bb )^\hodge = 0$, recovering~Equation \eqref{eq:diff-max-2}. In terms of components, the spatial Hodge complement directly transforms a scalar into a trivector with a unique component $\ebf_{123}$, recovering the homogeneous Maxwell equation in~Equation \eqref{eq:diff-max-2}.


We move on to the inhomogeneous Maxwell equations. We compute the interior derivative $\deltabf \lintprod \Fb$,
\begin{align}
\deltabf \lintprod \Fb &= (-\partial_0 \ebf_0 + \nabla ) \lintprod (\Fb _\Eb + \Fb _\Bb) 	\nonumber \\
&= -\partial_0 \Eb -\partial_0 \ebf_0 \lintprod \Bb^\hodge + \ebf_0 \nabla \cdot \Eb + \nabla \lintprod \Bb ^\hodge  \nonumber \\
&= -\partial_0 \Eb + \ebf_0 \nabla \cdot \Eb + \nabla \lintprod \Bb ^\hodge,
\end{align}
since $\ebf_0 \lintprod \Bb^\hodge = 0$. 
The interior derivative $\deltabf \lintprod \Fb$ gives a space--time vector with 4 components, while the inhomogeneous Maxwell equations are a scalar, Equation \eqref{eq:diff-max-1}, and a spatial vector, Equation \eqref{eq:diff-max-4}. 

We can verify that the scalar equation turns out to be given by the vector component $\ebf_{0}$ of $\deltabf \lintprod \Fb$, while the spatial vector equation is given by the spatial vector components $\ebf_{1}$, $\ebf_{2}$, and $\ebf_{3}$ of $\deltabf \lintprod \Fb$. Indeed, if we match this expression with the current density vector $\Jb$, then the time component $\ebf_0$ of $\deltabf \lintprod \Fb$ gives~Equation \eqref{eq:diff-max-1}. 
Selecting the space components of $\deltabf \lintprod \Fb$, the  differential equation is
\begin{equation}
-\partial_0 \Eb  + \nabla \lintprod \Bb ^\hodge = \jb ,
\end{equation}
which, using the relation $\nabla \lintprod \Bb ^\hodge = \nabla \times \Bb$ can be written as~Equation \eqref{eq:diff-max-4}.

\subsection{Equivalence of the Integral Form of Maxwell Equations}
After studying the exterior-calculus differential formulation of Maxwell equations, we  recover the standard  integral formulation. Applying the Stokes Theorem~\ref{theo:circ} to~Equation \eqref{eq:max-diff-homo}, we find that \textit{the circulation of the bivector field $\Fb$ along the boundary of any three-dimensional space--time volume ${\mathcal V}^3$ is zero}:
\begin{equation}\label{eq:86}
\int_{\partial\mathcal{V}^{3}}\drm^{2}\xb\cdot\Fb = \int_{\mathcal{V}^{3}}\drm^{3}\xb\cdot(\deltabf\wedge \Fb) = 0.
\end{equation}
At this point,~Equation \eqref{eq:86} is a scalar equation and we obtain the pair of homogeneous Maxwell equations by considering two different hypersurfaces $\mathcal{V}^{3}$.

First, let the domain ${\mathcal V}^3=V$ contain only spatial coordinates. There are no tangential components to $V$ with time indices and the contribution of $\Fb_\Eb$ to the circulation of $\Fb$ over $\partial\mathcal{V}^{3}$ in~Equation \eqref{eq:86} is zero, i.e.,
\begin{align}\label{eq:fB0}
\int_{\partial V}\drm^{2}\xb\cdot\Fb &= \int_{\partial V}\drm^{2}\xb\cdot\Fb_\Bb.
\end{align}
Using that $\ub\cdot\wb = \ub^{\hodgeinv}\cdot\wb^{\hodgeinv}$ for any vectors $\ub$, $\wb$, and therefore $\drm^{2}\xb\cdot\Fb = \drm^{2}\xb^{\hodgeinv}\cdot\Fb_\Bb^{\hodgeinv}$ and the definition $\Fb_\Bb = \Bb^\hodge$, the integral in the right-hand side of~Equation \eqref{eq:fB0} becomes
\begin{align}
\int_{\partial V}\drm^{2}\xb\cdot\Fb_\Bb &= \int_{\partial V}\drm^{2}\xb^{\hodgeinv}\cdot\Bb 
\nonumber \\
 &= \int_{\partial V}\drm\Sb\cdot\Bb, \label{eq:fB2}
\end{align}
where we used~Equation \eqref{eq:24} to write the last surface integral. Substituting~Equation \eqref{eq:fB2} back into~Equation \eqref{eq:86} gives the Gauss law for the magnetic field \cite{Feynman} (pp.~1-5--1-9).

Let now $\mathcal{V}^{3}$ be a time-space domain $(t_0,t_1)\times S$, where $S$ is a two-dimensional spatial surface. With no real loss of generality we assume that $S$ lies on the $\ebf_1\wedge\ebf_2$ plane.  Its boundary $\partial\mathcal{V}^{3}$ is the union of the sets $(t_0,t_1)\times \partial S$, $t_0\times S$ and $t_1\times S$. For the first set, we choose $\ebf_\perp$ as the vector normal to $\partial S$ pointing outwards on the plane defined by $S$ and $\ebf_\myparallel=\ebf_0\wedge \ebf_{\partial S}$, where $\ebf_{\partial S}$ is a vector tangent to $\partial S$ with a counterclockwise orientation, so that $\ebf_\perp\wedge\ebf_\myparallel=-\ebf_{012}$. Further, since $\ebf_\myparallel$ is a time-space bivector, the contribution of $\Fb_\Bb$ to the circulation of $\Fb$ over this first set in~Equation \eqref{eq:86} is zero, and
\begin{equation}
	\int_{(t_0,t_1)\times \partial S}\drm^{2}\xb\cdot\Fb = -\int_{(t_0,t_1)\times \partial S} \drm^2 \xb_\myparallel \cdot \Fb_\Eb.
	\label{eq:90} 
\end{equation}
Writing the differential vector as $\drm^2 \xb_\myparallel=\drm t \drm x \ebf_{0x}$, parameterizing the line integral over the boundary $\partial S$ by the variable $x$ with unit vector $\ebf_{0x}$, and using that $\ebf_{0x}\cdot \Fb_\Eb  = -\ebf_{x}\cdot \Eb$ and therefore $\drm x\,\ebf_{0x}\cdot \Fb_\Eb  = -\drm \xb\cdot \Eb$, the integral of the right-hand side of~Equation \eqref{eq:90} becomes
\begin{equation}
	\int_{t_0}^{t_1} \drm t\int_{\partial S} \drm \xb\cdot \Eb.
	\label{eq:91}
\end{equation}

For the second and third sets the normal vector to the integration surface pointing outwards are $\ebf_\perp=-\ebf_0$ and $\ebf_\perp=\ebf_0$ respectively. Since $\ebf_\myparallel$ is a space-space bivector in both cases, then the contribution of $\Fb_\Eb$ to the circulation is zero. We express the circulations of $\Fb_\Bb$ as fluxes of $\Bb$ and surface integrals as done in~Equation \eqref{eq:fB2}.
Using these observations the integral for the circulation of $\Fb$ over these two sets in~Equation \eqref{eq:86} is given by
\begin{align}
\int_{t_0\times S}\drm^{2}\xb\cdot\Fb + \int_{t_1\times S}\drm^{2}\xb\cdot\Fb &=  - \int_{S} \drm\Sb \cdot\Bb(t_0) + \int_{S} \drm\Sb \cdot\Bb(t_1).
\label{eq:92}
\end{align}

Combining~Equations \eqref{eq:91} and \eqref{eq:92} in~Equation \eqref{eq:86} we recover the integral over time of the so called Faraday law \cite{Feynman} (pp.~17-1--17-2). Equivalently, taking the time derivative recovers the usual Faraday law, namely 
\begin{align}
\int_{\partial S} \drm \xb\cdot \Eb + \partial_t\int_{S} \drm\Sb\cdot\Bb = 0.
\end{align}

In regard to the inhomogeneous Maxwell equations, applying the Stokes Theorem~\ref{theo:flux} to~Equation \eqref{eq:max-diff-nonhomo}, we find that {\em the flux of the bivector field $\Fb$ across 
the boundary of any three-dimensional space--time volume is equal to the flux of the current density $\Jb$ across the three-dimensional space--time volume}:
\begin{equation}\label{eq:93}
\int_{\partial\mathcal{V}^{3}}\drm^{2}\xb^{\hodgeinv}\cdot\Fb = 
	\int_{\mathcal{V}^{3}}\drm^{3}\xb^{\hodgeinv}\cdot(\deltabf\lintprod \Fb) = \int_{\mathcal{V}^{3}}\drm^{3}\xb^{\hodgeinv}\cdot\Jb .
\end{equation}
As with the homogeneous Maxwell equations, the scalar Equation \eqref{eq:93} yields the inhomogeneous Maxwell equations by considering two different hypersurfaces $\mathcal{V}^{3}$.

First, let the integration domain $\mathcal{V}^{3}$ be a spatial volume $V$. Since there are no normal components to $V$ with space indices only, the contribution of $\Fb_\Bb$ to the flux is zero so that~Equation \eqref{eq:93} becomes
\begin{align}\label{eq:78}
\int_{\partial V}\drm^{2}\xb^{\hodgeinv}\cdot\Fb_{\Eb} &= \int_{V}\drm^{3}\xb^{\hodgeinv}\cdot\Jb.
\end{align}
From the definition of inverse Hodge complement in~Equation \eqref{eq:hodge-inv-transf}, we write the differential vectors
\begin{gather}
\drm^{2}\xb^{\hodgeinv} = - \sum_{I, i\notin I}\drm^{2}x_{I} \sigma(0i,I)\ebf_{0i}
\label{eq:96}\\
\drm^{3}\xb^{\hodgeinv} = - \drm V  \ebf_{0}.
\end{gather}
Plugging these expressions in~Equation \eqref{eq:78}, using the definitions of $\Fb_\Eb$ and $\Jb$, and computing the dot products on both sides of the equality, we obtain that~Equation \eqref{eq:78} simplifies as  
\begin{gather}
\label{eq:98a}
-\int_{\partial V}\Biggl(\sum_{I, i\notin I}\drm^{2}x_{I} \sigma(0i,I)\ebf_{0i}\Biggr)\cdot \Biggl(\sum_{j}E_j\ebf_{0j}\Biggr) = -\int_{V}\drm V  \ebf_{0} \cdot (\rho\ebf_0+\jb)\\
\int_{\partial V}\sum_{I, i\notin I}\drm^{2}x_{I} \sigma(i,I)E_i = \int_{V}\drm V \rho\label{eq:82}\\
\int_{\partial V}\drm\Sb\cdot \Eb = \int_{V}\drm V\rho .\label{eq:83}
\end{gather}	
In~Equation \eqref{eq:82} we used that $\sigma(0i,I) = \sigma(i,I)$ and in~Equation \eqref{eq:83} we used that $\sum_{I, i\notin I}\drm^{2}x_{I} \sigma(i,I)E_i = \drm^{2}\xb^{\hodgeinv}\cdot\Eb$. Since the Hodge complement is over space, the result is a surface integral with positive orientation as in~Equation \eqref{eq:24}. We recovered in~Equation \eqref{eq:83} the Gauss law for the electric field \cite{Feynman} (pp.~4-7--4-9).

For $\mathcal{V}^{3} = (t_0,t_1)\times S$ where $S$ is a two-dimensional surface lying on the $\ebf_1\wedge\ebf_2$ plane, the boundary $\partial\mathcal{V}^{3}$ is the union of the sets $(t_0,t_1)\times \partial S$, $t_0\times S$ and $t_1\times S$. For the first set, since $\drm^{2}\xb^{\hodgeinv}$ has no time components, the contribution of $\Fb_\Eb$ to this set is zero, that is
\begin{equation}
	\int_{(t_0,t_1)\times \partial S}\drm^{2}\xb^{\hodgeinv}\cdot\Fb = \int_{(t_0,t_1)\times \partial S} \drm^{2}\xb^{\hodgeinv} \cdot \Fb_\Bb.
	\label{eq:101} 
\end{equation}
As in the homogeneous case, we choose $\ebf_\perp$ as the vector normal to $\partial S$ pointing outwards on the plane defined by $S$ and $\ebf_\myparallel=\ebf_0\wedge\ebf_{\partial S}$, where $\ebf_{\partial S}$ is a vector tangent to $\partial S$ with a counterclockwise orientation, such that $\ebf_\perp\wedge\ebf_\myparallel=-\ebf_{012}$ introduces a change of sign. Expressing $\drm^{2}\xb_\myparallel^{\hodgeinv}$ and $\Fb_\Bb$ in the canonical basis, defining $I=(0,i)$ so that $I^c$ contains only space indexes, and using that $\ebf_{I^c}\cdot\ebf_{i^c} = 1$ and $\sigma(I^c,I) = \sigma(I^c,0,i) = \sigma(0,I^c,i) = \sigma(I^c,i)$, we obtain that Equation \eqref{eq:101} simplifies to
\begin{align}
\int_{(t_0,t_1)\times \partial S}\drm^{2}\xb^{\hodgeinv}\cdot\Fb &= \int_{(t_0,t_1)\times \partial S}\left(\sum_I \drm x_I\sigma(I^c,I)\ebf_{I^c}\right)\cdot \left(\sum_{i}B_i\ebf_{i^c}\sigma(i^c,i)  \right)  \nonumber \\
&= -\int_{t_0}^{t_1}\drm t\int_{\partial S}\drm x B_x  \nonumber \\
&= -\int_{t_0}^{t_1}\drm t\int_{\partial S}\drm \xb \cdot \Bb.\label{eq:88}
\end{align}

For the second and third sets we respectively choose $\ebf_\perp=-\ebf_0$ and $\ebf_\perp=\ebf_0$ pointing outside ${\mathcal V}^3$, implying that the contribution of $\Fb_\Bb$ is zero for this set as the inverse Hodge complement of $\ebf_\perp$ is a space vector. Expressing $\drm^{2}\xb^{\hodgeinv}$ in Equation \eqref{eq:96} and using similar steps as in Equations \eqref{eq:98a}--\eqref{eq:83}, the left-hand side of Equation \eqref{eq:93} over these two sets is given by
\begin{equation}
	\int_{t_0 \times \partial S}\drm^{2}\xb^{\hodgeinv}\cdot\Fb +\int_{t_1 \times \partial S}\drm^{2}\xb^{\hodgeinv}\cdot\Fb   = - \int_{S}\drm\Sb\cdot \Eb(t_0) +  \int_{S}\drm\Sb\cdot \Eb(t_1)
	\label{eq:max-part1}
\end{equation}

Finally, for the right-hand side of~Equation \eqref{eq:93}, we choose $\ebf_\perp$ as the vector normal to ${\mathcal V}^3$ pointing outside. Since  $\ebf_\myparallel=\ebf_{012}$ implies that $\ebf_\perp\wedge\ebf_\myparallel=-\ebf_{0123}$, we obtain that
\begin{align}
\int_{{\mathcal V}^3}\drm^{3}\xb^{\hodgeinv}\cdot\Jb = -\int_{t_0}^{t_1}\drm t\int_{S}\drm\Sb\cdot\jb. \label{eq:max-part3}
\end{align}

We have thusly recovered the integral form of the Ampere-Maxwell equation \cite{Feynman} (p.~18-1--18-4) integrated over the time interval $(t_0,t_1)$ by combining~Equations \eqref{eq:88}--\eqref{eq:max-part3} into Equation \eqref{eq:93}, that is
\begin{align}
\int_{t_0}^{t_1}\drm t\int_{\partial S}\drm \xb \cdot \Bb &= \int_{t_0}^{t_1}\drm t\int_{S}\drm\Sb\cdot\jb +  \int_{S}\drm\Sb\cdot \Eb(t_1)- \int_{S}\drm\Sb\cdot \Eb(t_0) .
\end{align}

\section{Summary}
In this paper, we aimed at showing how exterior calculus provides a tool merging the simplicity and intuitiveness of standard vector calculus with the power of tensors and differential forms. Set in the context of a general space--time algebra with multiple space and time components, we provided the basic concepts of exterior algebra and calculus, such as multivectors, wedge product and interior products, with a distinction between left and right products, Hodge complement, and exterior and interior derivatives. While a space--time with multiple time coordinates leads to several issues from the physical point of view  \cite{Velev}, we did not deal with these problems as this paper focuses on the mathematical constructions.
We also defined oriented integrals, with two important examples being the flux and circulation of grade $m$-vector fields as integrals of the normal and tangent components of the field to a hypersurface respectively. These operations extend the standard circulation of a vector field as a line integral and the flux of a vector field as a surface integral in three dimensions to any number of dimensions and any vector grade.

Armed with the theory of differential forms, we proved two exterior-calculus Stokes theorems, one for the circulation and one for the flux, that generalize the Kelvin-Stokes, Gauss and Green theorems. We saw that the flux of the curl of a vector field in three dimensions across a surface is also the circulation of the exterior derivative of the vector field along that surface. In exterior calculus, these Stokes theorems hold for any number of dimensions and any vector grade and are simply expressed in terms of the exterior and interior derivatives for the circulation and flux respectively. 

As an application of our tools, we showed how to recover the classical laws of electromagnetism, Maxwell equations and Lorentz force, from a exterior-calculus formalism in relativistic space--time with one temporal and three spatial dimensions. The electromagnetic field is described by a bivector field with six components, closely related to Faraday's antisymmetric tensor, containing both electric and magnetic fields. The differential form of Maxwell equations relates the exterior derivative of the bivector field with the zero trivector and the interior derivative of the field with the current density vector. In the integral form, these equations correspond to the statements that the circulation of the bivector field along the boundary of any three-dimensional space--time volume is zero, and that the flux of the bivector field across the boundary of any three-dimensional space--time volume is equal to the flux of the current density across the same space--time volume. 



\appendix
\renewcommand\thesection{\Alph{section}}

\section{Proofs of product identities}	\label{sec-app-A}
In this appendix, we verify the relations about interior products introduced in Section~\ref{sec-EC-basics}.
\subsection{Relation between left and right interior products} \label{sec-app-left-right}
We now prove the formula
\begin{equation} \label{eq:left-right-int-prod}
\ebf_I \lintprod \ebf_J = \ebf_J \rintprod \ebf_I (-1)^{\len{I}(\len{J}-\len{I})} ,
\end{equation}
relating left and right interior products.
For two lists $I$ and $J$, we have 
\begin{gather}
\ebf_I \lintprod \ebf_J = \Delta_{I,I}\sigma\bigl(J\setminus I, I\bigr)\ebf_{J \setminus I} , \\
\ebf_J \rintprod \ebf_I = \Delta_{I,I}\sigma\bigl(I, J\setminus I \bigr)\ebf_{J \setminus I} ,
\end{gather}
where we assumed that $I\subseteq J$ with no loss of generality and used that $\varepsilon(I,J^c)^c = J\setminus I$ in this case.
The only difference between the expressions lies in the signatures, that are related by setting $A=J\setminus I$ and $B=I$ in the following lemma.

\begin{lemma} \label{lemma:reverse}
Given two arbitrary lists $A$ and $B$, of length $\len{A}$ and $\len{B}$ respectively, then the permutations sorting the concatenated lists $(A,B)$ and $(B,A)$ satisfy the formula
\begin{equation}
\label{eq:sigma-reverse-binom}
\sigma(A,B) = \sigma(B,A) (-1)^{\len{A}\len{B}}.
\end{equation}
\end{lemma}
\begin{proof}
Given a list $A$, let $\widebar{A}$ be the reversed list, namely the list where the order of all the elements is reversed. Counting the number of position jumps needed to reverse the list, we obtain the signature of this reversing operation as
\begin{equation}\label{eq:sign_reversal}
\sigma_\text{r}(A)= \sigma_\text{r}(\widebar{A}) = (-1)^{\len{A}-1+\len{A}-2+\dotsc+1} = (-1)^{\frac{\len{A}(\len{A}-1)}{2}}.
\end{equation}

The proof is based on the identity between two different ways of rearranging the concatenated list $(A,B)$ into the ordered list $\varepsilon(A,B)$, as depicted in~\figurename~\ref{fig:int-equiv}. 
\begin{figure}[htb]
\centering
\begin{subfigure}[t]{0.45\textwidth}
	\centering
\begin{tikzpicture}
\draw [blue, thick] (-3,5)node[black] {$\vert$} -- (1,5) node[black] {$\vert$};
\node[blue, above] at (-1,5) {$A$};
\draw [red, thick] (1,5)node[black] {$\vert$} -- (3,5) node[black] {$\vert$};
\node[red, above] at (2,5) {$B$};
\draw [white, thick] (0,3) -- (0,3) node[white]{A}; 
\draw [white, thick] (0,2) -- (0,2) node[white]{A};; 
\draw [thick] (-3,4)node[black] {$\vert$} -- (3,4) node[black] {$\vert$};
\node[above] at (0,4) {$\varepsilon(A,B)$};
\end{tikzpicture}
\end{subfigure}
\begin{subfigure}[t]{0.45\textwidth}
\centering
\begin{tikzpicture}
\draw [blue, thick] (-5,5)node[black] {$\vert$} -- (-1,5) node[black] {$\vert$};
\node[blue, above] at (-3,5) {$A$};
\draw [red, thick] (-1,5)node[black] {$\vert$} -- (1,5) node[black] {$\vert$};
\node[red, above] at (0,5) {$B$};
\draw [blue, thick] (-3,4) node[black] {$\vert$} -- (1,4) node[black] {$\vert$};
\node[blue, above] at (-1,4) {$\widebar{A}$};
\draw [red, thick] (-5,4)node[black]{$\vert$} -- (-3,4) node[black] {$\vert$};
\node[red, above] at (-4,4) {$\widebar{B}$};
\draw [blue, thick] (-3,3) node[black] {$\vert$} -- (1,3) node[black] {$\vert$};
\node[blue, above] at (-1,3) {$A$};
\draw [red, thick] (-5,3)node[black]{$\vert$} -- (-3,3) node[black] {$\vert$};
\node[red, above] at (-4,3) {$B$};
\draw [thick] (-5,2)node[black] {$\vert$} -- (1,2) node[black] {$\vert$};
\node[above] at (-2,2) {$\varepsilon(A,B)$};
\end{tikzpicture}
\end{subfigure}
\caption{Visual aid for the relation between $\sigma(A,B)$ and $\sigma(B,A)$.}	\label{fig:int-equiv}
\end{figure}
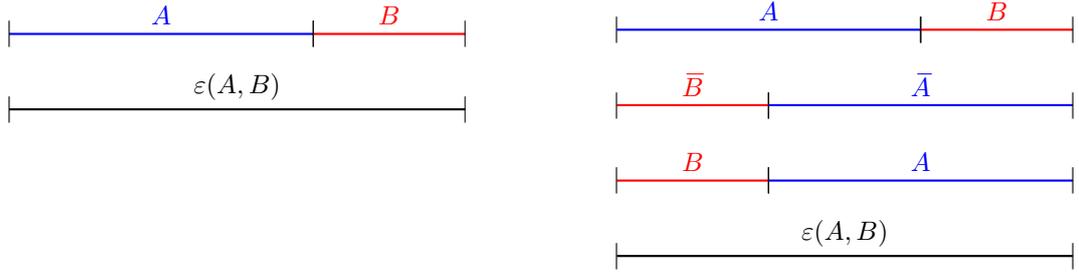

First, in the left column of~\figurename~\ref{fig:int-equiv} we depict how a single permutation with signature $\sigma(A,B)$ orders the list $(A,B)$. In the right column of~\figurename~\ref{fig:int-equiv} we depict how a different series of permutations achieves the same result. We start by reversing the concatenated list $(A,B)$, an operation with signature $\sigma_\text{r}(\widebar{B},\widebar{A})$. Then, we separately partially reverse the lists $\widebar{B}$ and $\widebar{A}$, operations with respective signatures $\sigma_\text{r}(\widebar{B})$ and $\sigma_\text{r}(\widebar{A})$. A final permutation with signature $\sigma(B,A)$ orders the list $(B,A)$ into $\varepsilon(A,B)$. Since the signature of a composition of permutations is the product of the signatures, we obtain that 
\begin{equation} \label{eq:sigma-reverse}
\sigma(A,B) = \sigma_\text{r}(\widebar{B},\widebar{A})\sigma_\text{r}(\widebar{A})\sigma_\text{r}(\widebar{B})\sigma(B,A).
\end{equation}

Using~Equation \eqref{eq:sign_reversal} in every $\sigma_\text{r}$ in~Equation \eqref{eq:sigma-reverse} and carrying out some simplifications yields~Equation \eqref{eq:sigma-reverse-binom}.
\end{proof}

\subsection{Relation between interior and exterior products} \label{sec-app-int-ext}
 We start with the expression for the left interior product~Equation \eqref{eq:left-int-equiv}. From~Equations \eqref{eq:hodge-transf} and~\eqref{eq:hodge-inv-transf}, we compute
\begin{align}
 \bigl( \ebf_I \wedge \ebf_J^{\hodge} \bigr)^{\hodgeinv} 
 &= \bigl( \Delta_{J,J}\sigma(J,J^c) \ebf_I \wedge \ebf_{J^c} \bigr)^{\hodgeinv}  \nonumber \\
  &= \Delta_{J,J} \Delta _{\varepsilon(I,J^c)^c, \varepsilon(I,J^c)^c} 
  \, \sigma(J,J^c) \sigma(I,J^c) \sigma(\varepsilon(I,J^c)^c, \varepsilon(I,J^c)) \, \ebf_{\varepsilon(I,J^c)^c} ,
\end{align}
and since $\Delta_{\varepsilon(I,J^c)^c, \varepsilon(I,J^c)^c} = \Delta _{J\setminus I,J\setminus I} $, we can conclude that
$
\Delta_{J,J} \Delta_{\varepsilon(I,J^c)^c, \varepsilon(I,J^c)^c}
= \Delta_{I,I}
$.
If we now compare the result with~Equation \eqref{eq:left-int-prod}, we need just to verify the identity
\begin{equation} \label{eq:sigma-composition}
\sigma\bigl(\varepsilon(I,J^c)^c,I\bigr) =
 \sigma(J,J^c) \sigma(I,J^c) \sigma(\varepsilon(I,J^c)^c, \varepsilon(I,J^c)) \,,
\end{equation}
or equivalently
\begin{equation}
\label{eq:sigmas-left-scalar}
 \sigma\bigl(\varepsilon(I,J^c)^c,I\bigr)  \sigma(J,J^c)=
\sigma(\varepsilon(I,J^c)^c, \varepsilon(I,J^c))
 \sigma(I,J^c) \,.
\end{equation}
The left-hand side of~Equation \eqref{eq:sigmas-left-scalar}
corresponds to taking the sets $\varepsilon(I,J^c)^c = J \setminus I$, $I$ and $J^c$, in this order, and then merging and sorting $J\setminus I$ with $I$ and then merging and sorting the resulting set $J$ with $J^c$, as shown in the left column of \figurename~\ref{fig:left-scalar}.
On the right-hand side, we start we the same three lists, but we first merge and sort $I$ with $J^c$, and then we get the whole list by merging and sorting the result with $J\setminus I$, as represented in the right column of \figurename~\ref{fig:left-scalar}. 
Thus, starting from the three sets and rearranging them in different ways, we get the same final ordered list, and since the signatures of the left-hand side and right-hand side are the same, and
Equation \eqref{eq:sigmas-left-scalar} is proved. As a consequence,~Equation \eqref{eq:left-int-equiv} is verified.
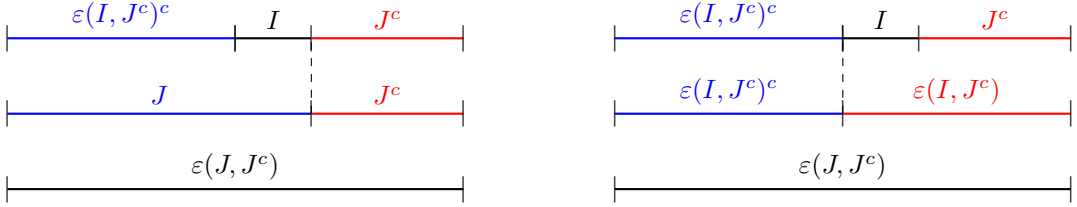
\begin{figure}[htb]
\centering
	\begin{subfigure}[t]{0.45\textwidth}
	\centering
\begin{tikzpicture}
\draw [blue, thick] (-3,3)node[black] {$\vert$} -- (1,3) node[black] {$\vert$};
\node[blue, above] at (-1,3) {$J$};
\draw [red, thick] (1,3) -- (3,3) node[black] {$\vert$};
\node[red, above] at (2,3) {$J^c$};
\draw [blue, thick] (-3,4) node[black] {$\vert$} -- (0,4) node[black] {$\vert$};
\node[blue, above] at (-1.5,4) {$\varepsilon(I,J^c)^c$};
\draw [thick] (0,4) node[black] {$\vert$} -- (1,4) node[black] {$\vert$};
\node[above] at (0.5,4) {$I$};
\draw [red, thick] (1,4) -- (3,4) node[black] {$\vert$};
\node[red, above] at (2,4) {$J^c$};
\draw[dashed] (1,3) -- (1,4);
\draw [thick] (-3,2)node[black] {$\vert$} -- (3,2) node[black] {$\vert$};
\node[above] at (0,2) {$\varepsilon(J, J^c)$};
\end{tikzpicture}
\end{subfigure}
\centering
	\begin{subfigure}[t]{0.45\textwidth}
	\centering
\begin{tikzpicture}
\draw [blue, thick] (-3,3) node[black] {$\vert$} -- (0,3) node[black] {$\vert$};
\node[blue, above] at (-1.5,3) {$\varepsilon(I,J^c)^c$};
\draw [red, thick] (0,3) -- (3,3) node[black] {$\vert$};
\node[red, above] at (1.5,3) {$\varepsilon(I,J^c)$};
\draw [blue, thick] (-3,4) node[black] {$\vert$} -- (0,4) node[black] {$\vert$};
\node[blue, above] at (-1.5,4) {$\varepsilon(I,J^c)^c$};
\draw [thick] (0,4) node[black] {$\vert$} -- (1,4) node[black] {$\vert$};
\node[above] at (0.5,4) {$I$};
\draw [red, thick] (1,4) -- (3,4) node[black] {$\vert$};
\node[red, above] at (2,4) {$J^c$};
\draw[dashed] (0,3) -- (0,4);
\draw [thick] (-3,2)node[black] {$\vert$} -- (3,2) node[black] {$\vert$};
\node[above] at (0,2) {$\varepsilon(J, J^c)$};
\end{tikzpicture}
\end{subfigure}
\caption{Visual aid for the permutations in~Equation \eqref{eq:sigmas-left-scalar}.}
\label{fig:left-scalar}
\end{figure}


Afterwards, we prove the formula for the right interior product~Equation \eqref{eq:right-int-equiv}.
Using~Equations \eqref{eq:hodge-transf} and~\eqref{eq:hodge-inv-transf}, we write
\begin{align}
\bigl( \ebf_I^{\hodgeinv} \wedge \ebf_J \bigr)^{\hodge} &= 
\bigl( \Delta_{I^c,I^c}\sigma(I^c, I)\ebf_{I^c} \wedge \ebf_{J} \bigr)^{\hodge} \nonumber \\ &=
\Delta_{I^c,I^c}\sigma(I^c, I) \sigma(I^c, J)\, \ebf_{\varepsilon(I^c,J)}^\hodge	\nonumber \\
&= \Delta_{I^c,I^c}\sigma(I^c, I) \sigma(I^c, J) \Delta_{\varepsilon(I^c,J), \varepsilon(I^c,J)} \sigma(\varepsilon(I^c,J), \varepsilon(I^c,J)^c)  \, \ebf_{\varepsilon(I^c,J)^c}.
\end{align}
Using that $\Delta_{\varepsilon(I^c,J), \varepsilon(I^c,J)}\Delta_{I^c, I^c}=\Delta_{J,J}$, in order to prove the validity of~Equation \eqref{eq:right-int-equiv}, we need to prove the relation 
\begin{equation}
 \sigma(J,\varepsilon(I^c,J)^c) =
\sigma(I^c,I)  \sigma(I^c,J)
 \sigma(\varepsilon(I^c,J),\varepsilon(I^c,J)^c) \,.
\end{equation}
We can prove it applying Lemma~\ref{lemma:reverse} to obtain the expression~Equation \eqref{eq:sigma-composition}, or by following the same procedure as before, paying attention to the difference that now list $J$ is included in $I$.

\subsection{Triple mixed product}
\label{sec-app-mix-prod}
Given two 1-vectors $\ub$ and $\vb$ and a $r$-vector $\wb$, we prove the relation
\begin{equation}
\ub\lintprod(\vb\wedge\wb) = (-1)^{r}(\ub\cdot\vb)\wb + \vb\wedge(\ub\lintprod\wb). \label{eq:app-triple}
\end{equation}
\begin{proof}
We start by evaluating $\ub\lintprod(\vb\wedge\wb)$ explicitly, separating terms $i=j$ and $i\ne j$, namely
\begin{equation}
\ub\lintprod(\vb\wedge\wb) = \sum _{\substack{i,j,I\\j\notin I,i\in I}} \Delta_{i,i} u_i v_j w_I \sigma(j,I) \sigma(I+j\setminus i,i) \ebf_{I+j\setminus i} + \sum _{\substack{i,I\\i\in I}} \Delta_{i,i} u_i v_i w_I \sigma(i,I) \sigma(I,i) \ebf_{I},
\end{equation}
then, using $\sigma(i,I) \sigma(I,i)=(-1)^r$ and adding and removing a term
$\displaystyle (-1)^r\sum _{\substack{i,I\\i\notin I}} \Delta_{i,i} u_i v_i w_I \ebf_{I}$, we get
\begin{equation}
\ub\lintprod(\vb\wedge\wb) = \sum _{\substack{i,j,I\\i\in I,j\notin I\setminus i}} \Delta_{i,i} u_i v_j w_I \sigma(j,I) \sigma(I+j\setminus i,i) \ebf_{I+j\setminus i} + (-1)^r \sum _{i,I} \Delta_{i,i} u_i v_i w_I \ebf_{I}.
\end{equation}
More concretely, the left-hand side $\ub\lintprod(\vb\wedge\wb)$ is given by 
\begin{align}
&\sum _{\substack{i,j,I\\j\notin I,i\in I}} \Delta_{i,i} u_i v_j w_I \sigma(j,I) \sigma(I+j\setminus i,i) \ebf_{I+j\setminus i}
- (-1)^r\sum _{\substack{i,I\\i\notin I}} \Delta_{i,i} u_i v_i w_I \ebf_{I}	\nonumber \\
&= \sum _{\substack{i,j,I\\j\notin I,i\in I}} \Delta_{i,i} u_i v_j w_I \sigma(j,I) \sigma(I+j\setminus i,i) \ebf_{I+j\setminus i}
- \sum _{\substack{i,j,I\\j=i,j\notin I\setminus i,i\in I}} \Delta_{i,i} u_i v_j w_I \sigma(j,I) \sigma(I+j\setminus i,i) \ebf_{I+j\setminus i}	\nonumber \\
&=\sum _{\substack{i,j,I\\i\in I,j\notin I\setminus i}} \Delta_{i,i} u_i v_j w_I \sigma(j,I) \sigma(I+j\setminus i,i) \ebf_{I+j\setminus i}.\label{eq:A27}
\end{align}

Similarly, we evaluate the right-hand side $\vb\wedge(\ub\lintprod\wb)$ as 
\begin{equation}\label{eq:A28}
\vb\wedge(\ub\lintprod\wb) = 
 \sum _{\substack{i,j,I\\i\in I,j\notin I\setminus i}} \Delta_{i,i} u_i v_j w_I \sigma(I\setminus i,i) \sigma(j,I \setminus i) \ebf_{I+j\setminus i} .
\end{equation}
Comparing~Equations \eqref{eq:A27} and~\eqref{eq:A28}
, it remains to prove the equality,
and now we prove the equality
\begin{equation}
\label{eq:app-prove-3prod}
 \sigma(j,I) \sigma(I+j\setminus i,i) =\sigma (I\setminus i,i)\sigma(j,I\setminus i).
\end{equation}
We rewrite Equation \eqref{eq:app-prove-3prod} multiplying both sides for $ \sigma(j,I)\sigma(j,I\setminus i)$ so that we obtain
\begin{equation}
\sigma(j,I\setminus i) \sigma(I+j\setminus i,i) =\sigma (I\setminus i,i) \sigma(j,I)
\end{equation}
which we verify with the help of \figurename~\ref{fig:forprod}. 
On the left column, we first merge $j$ with $I\setminus i$ and then the resulting list with $i$. On the right columns, the permutations first join $I\setminus i$ and $i$ and the resulting $I$ is then merged with $j$, getting the same result in both sides of the relation. 
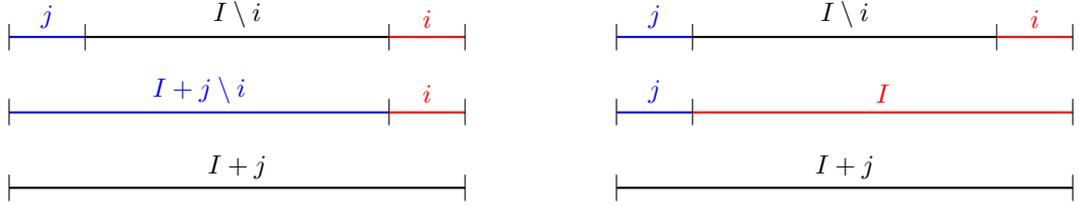
\begin{figure}[htb]
\centering
	\begin{subfigure}[t]{0.45\textwidth}
	\centering
\begin{tikzpicture}
\draw [blue, thick] (-3,4) node[black] {$\vert$} -- (-2,4) node[black] {$\vert$};
\node[blue, above] at (-2.5,4) {$j$};
\draw [thick] (-2,4)  -- (2,4) node[black] {$\vert$};
\node[above] at (0,4) {$I\setminus i$};
\draw [red, thick] (2,4) -- (3,4) node[black] {$\vert$};
\node[red, above] at (2.5,4) {$i$};
\draw [blue, thick] (-3,3) node[black] {$\vert$} -- (2,3) node[black] {$\vert$};
\node[blue, above] at (-0.5,3) {$I+j\setminus i$};
\draw [red, thick] (2,3) -- (3,3) node[black] {$\vert$};
\node[red, above] at (2.5,3) {$i$};
\draw [thick] (-3,2)node[black] {$\vert$} -- (3,2) node[black] {$\vert$};
\node[above] at (0,2) {$I+j$};
\end{tikzpicture}
\end{subfigure}
\centering
	\begin{subfigure}[t]{0.45\textwidth}
	\centering
\begin{tikzpicture}
\draw [blue, thick] (-3,4) node[black] {$\vert$} -- (-2,4) node[black] {$\vert$};
\node[blue, above] at (-2.5,4) {$j$};
\draw [thick] (-2,4)  -- (2,4) node[black] {$\vert$};
\node[above] at (0,4) {$I\setminus i$};
\draw [red, thick] (2,4) -- (3,4) node[black] {$\vert$};
\node[red, above] at (2.5,4) {$i$};
\draw [blue, thick] (-3,3) node[black] {$\vert$} -- (-2,3) node[black] {$\vert$};
\node[blue, above] at (-2.5,3) {$j$};
\draw [red, thick] (-2,3) -- (3,3) node[black] {$\vert$};
\node[red, above] at (0.5,3) {$I$};
\draw [thick] (-3,2)node[black] {$\vert$} -- (3,2) node[black] {$\vert$};
\node[above] at (0,2) {$I+j$};\end{tikzpicture}
\end{subfigure}
\caption{Visual aid for the identity $\sigma(j, I\setminus i) \sigma(I+j \setminus i,i) = \sigma(I\setminus i,i) \sigma(j,I)$.}
\label{fig:forprod}
\end{figure}
Thus, we can write
\begin{equation}
\ub\lintprod(\vb\wedge\wb) = \sum _{\substack{i,j,I\\i\in I,j\notin I\setminus i}} \Delta_{i,i} u_i v_j w_I \sigma(I\setminus i,i) \sigma(j,I \setminus i) \ebf_{I+j\setminus i}
+ (-1)^r \left( \sum _{i} \Delta_{i,i} u_i v_i \right)\left(\sum _{I}  w_I \ebf_{I}\right),
\end{equation}
where we identify the term $(-1)^{r}(\ub\cdot\vb)\wb$, and finally conclude
\begin{equation}
\ub\lintprod(\vb\wedge\wb) - \vb\wedge(\ub\lintprod\wb) = 
(-1)^r (\ub \cdot \vb) \wb,
\end{equation}
which proves our initial formula.
\end{proof}

\bibliographystyle{IEEEtran}	
\bibliography{IC-bib-th.bib}

\end{document}